\documentclass[10pt,journal,cspaper,compsoc]{IEEEtran}
\usepackage[nocompress]{cite}
\usepackage[cmex10]{amsmath}
\usepackage{amssymb}
\usepackage{amsthm}

\usepackage{algpseudocode}

\makeatletter
\let\OldStatex\Statex
\renewcommand{\Statex}[1][3]{%
  \setlength\@tempdima{\algorithmicindent}%
  \OldStatex\hskip\dimexpr#1\@tempdima\relax}
\makeatother
\algnewcommand{\A}{\textbf{and}\space}
\algnewcommand{\Or}{\textbf{or}\space}
\algnewcommand{\Xor}{\textbf{xor}\space}

\usepackage[unicode=true,pdfusetitle,
 bookmarks=true,bookmarksnumbered=false,bookmarksopen=false,
 breaklinks=false,pdfborder={0 0 1},backref=false,colorlinks=false]
 {hyperref}
\usepackage[usenames,dvipsnames]{xcolor}
\hypersetup{colorlinks=true, linkcolor=Maroon, citecolor=OliveGreen, filecolor=magenta, urlcolor=Blue}

\global\long\def\w{\omega}
\global\long\def\st{\sqrt{2}}
\global\long\def\ve{\varepsilon}
\global\long\def\vea{\ve^{\ast}}
\global\long\def\th{\theta}
\global\long\def\dg{\dagger}
\global\long\def\d{\delta}
\global\long\def\ks{k^{\ast}}
\global\long\def\p{\phi}
\global\long\def\Rzp{R_z(\p)}
\global\long\def\Z{\mathbb{Z}}
\global\long\def\R{\mathbb{R}}

\global\long\def\T{\mathcal{T}}
\global\long\def\N{\mathcal{N}}
\global\long\def\U{\mathcal{U}}
\global\long\def\Um{\mathcal{\U}^{\min}}
\global\long\def\Zr{\mathbb{Z}\left[\w\right]}

\global\long\def\Zf{\mathbb{Z}[i, \frac{1}{\sqrt{2}}]}
\global\long\def\md{\,\mathrm{mod}\,}

\global\long\def\sde{\mathrm{sde}}

\global\long\def\sol{\partial}

\global\long\def\re{\mathrm{Re}}
\global\long\def\im{\mathrm{Im}}
\global\long\def\l{\left}
\global\long\def\r{\right}
\global\long\def\lf{\lfloor}
\global\long\def\rf{\rfloor}

\global\long\def\BFS{\mathrm{BFS}}

\global\long\def\Zw{\mathbb{Z}\left[\w\right]}
\global\long\def\Zs{\mathbb{Z}\left[\sqrt{2}\right]}

\global\long\def\ket#1{\left|#1\right\rangle }

\global\long\def\set#1#2{\l\{\l.#1\,\r|\,#2\r\}}
\global\long\def\sets#1#2{\l\{#1\,\l|\,#2\r.\r\}}

\newtheorem{thm}{Theorem}
\newtheorem{prob}{Problem}
\newtheorem{lem}{Lemma}
\newtheorem{prop}{Proposition}

\usepackage{pgfplots}
\pgfplotsset{compat=newest}
\usetikzlibrary{calc}

\begin{document}


\title{Practical approximation of single-qubit unitaries by single-qubit quantum Clifford and T circuits}

\author{Vadym~Kliuchnikov, Dmitri~Maslov, Michele~Mosca
\IEEEcompsocitemizethanks{
\IEEEcompsocthanksitem V.~Kliuchnikov is with Institute for Quantum Computing, and David R. Cheriton School of
Computer Science, University of Waterloo, Waterloo, Ontario, Canada\protect\\
E-mail: v.kliuchnikov@gmail.com
\IEEEcompsocthanksitem D.~Maslov is with National Science Foundation, Arlington, Virginia, USA\protect\\
E-mail: dmitri.maslov@gmail.com
\IEEEcompsocthanksitem M.~Mosca is with Institute for Quantum Computing, and Dept. of Combinatorics \& Optimization, University of Waterloo, Waterloo, Ontario, Canada, and Perimeter Insitute for Theoretical Physics, Waterloo, Ontario, Canada\protect\\
E-mail: michele.mosca@uwaterloo.ca
}
\thanks{}}

\markboth{}
{V.~Kliuchnikov \MakeLowercase{\textit{et al.}}: Practical approximation of single-qubit unitaries...}


\IEEEcompsoctitleabstractindextext{%
\begin{abstract}
We present an algorithm, along with its implementation that finds T-optimal approximations of single-qubit Z-rotations using quantum circuits consisting of Clifford and T gates.  Our algorithm is capable of handling errors in approximation down to size $10^{-15}$, resulting in optimal single-qubit circuit designs required for implementation of scalable quantum algorithms.  Our implementation along with the experimental results are available in the public domain.
\end{abstract}

}


\maketitle

\IEEEdisplaynotcompsoctitleabstractindextext

\IEEEpeerreviewmaketitle


\section{Introduction}
\IEEEPARstart{Q}{uantum} computing is a recent computing paradigm using the laws of quantum mechanics as a basis for computation.  The following two observations explain the interest in the study of this computing model. First, it has been shown that quantum algorithms can solve certain computational problems more efficiently than the best known classical algorithms.  The speed-up provided by quantum algorithms is sometimes quite significant, including superpolynomial for the well-known integer factorization problem (more generally, the hidden subgroup problem over Abelian groups; the original quantum algorithm is best known as Shor's algorithm).  Second, small quantum computations have already been demonstrated in experiments, and recent results in scaling and fault tolerance suggest the possibility of a full-scale quantum computation.  As a result, quantum computations may one day become a hardware platform capable of substantially speeding up certain computations in ways classical computation is 
believed to be incapable of.  

Much like any classical algorithm, a quantum algorithm needs to be implemented efficiently in order to gain maximal possible advantage from executing it.  Typically, a quantum algorithm is described in terms of high level procedures such as arithmetic operations (addition, multiplication, exponentiation) or special purpose transforms, such as the Quantum Fourier Transform (QFT).  These large transforms are then decomposed into high level logical gates, such as Toffoli, Fredkin, SWAP, arbitrary two-qubit gates, including controlled versions of the above, {\em etc.}, and finally broken down into circuits over elementary logical gates.  The set of the elementary logical gates allowed is dictated by the fault-tolerance techniques that limit the efficiency of implementing an arbitrary transformation.  Recent studies of fault tolerance techniques suggest that the fault-tolerant library should consist of Clifford (single-qubit Pauli, Hadamard, Phase, CNOT gates) and T logical gates, with the understanding that 
the T gate 
requires considerably more resources than any of the Clifford gates \cite{quant-ph/0504218, arXiv:0803.0272, bk:nc}.  Consequently, since recently, it has become widely accepted that the T-gate count/depth may serve as a good first-order approximation of the resource count required to physically implement a quantum circuit. 

In this paper we study the problem of the optimal single-qubit gate approximation by Clifford and T circuits.  Single-qubit gates arise in a variety of contexts within quantum algorithms, most notably, in the Quantum Fourier Transform (per \cite{arXiv:1206.5236}, controlled-Z rotations can be implemented by reducing them to Fredkin and one single-qubit gate; it is the single-qubit gate that requires approximation and consumes most resources), and quantum simulations \cite{chem,chem-small}.  Interestingly, in both cases the single qubit gates required are the rotations around the axis $Z$, which are the gates we approximate optimally in this paper.

A solution to the single-qubit circuit approximation problem in the form of a brute force search to find optimal circuits was suggested by Fowler \cite{F1}. However, brute force search appears to run out of classical computational resources for approximation error values below $10^{-4}$. The approximation precision can be improved using the Solovay-Kitaev algorithm \cite{bk:ksv,DN}. When using it the resulting circuit size scales as $O(\log^{3.97}(1/\ve))$ instead of the optimal scaling $O(\log(1/\ve))$. In contrast, our algorithm is capable of handling precision down to $10^{-15}$ and producing optimal results, and is thus suited for application to scalable quantum computing.  Next, an exact synthesis algorithm has been developed to synthesize unitaries over the ring $\Z[i, 1/\sqrt{2}]$~\cite{arXiv:1206.5236}.  This algorithm synthesizes circuits that are both T- and H-optimal.  However, it does not answer the question of how to efficiently approximate a single-qubit unitary whose elements lie outside the 
ring $\Z[i, 1/\sqrt{2}]$, and in most practical situations this is precisely the case.   We rely on this latter algorithm in our paper, as well as on the observation made in~\cite{arXiv:1206.5236} that finding an approximating 
circuit is as difficult as finding the approximating unitary.  More recently, \cite{arXiv:1212.0822} developed an algorithm for finding asymptotically optimal single-qubit circuit approximations by using a few ancillae.  While this means that the resource count asymptotics have thus been settled up to a constant factor in a scalable (polynomial in $\log{(1/\ve)}$ classical resources required to synthesize an asymptotically optimal quantum circuit) fashion, those constant factors matter in the actual implementations.  At approximately the same time as our original posting of this work, a new result appeared \cite{S}, that shows how to approximate single-qubit unitaries with an approximately $33\%$ overhead compared to the optimal results, and using no ancillae.  While the focus of \cite{S} appears to be on scaling to handle very tiny errors, our focus is on minimizing the quantum computing resources for implementation sizes of foreseeable practical importance.  Indeed, we obtain 
optimal circuits, meaning further simplification is impossible unless other additional resources are allowed.  We also highlight that a number of approaches have been developed in the literature that use additional resources in the form of ancillae, special states, classical feedback, or whose application results in a probabilistic success of having approximated a target unitary \cite{DCS,PS,BGS,J,bk:ksv,WK,J2}.  In contrast to those publications, our focus is on solving the basic version of the synthesis problem---the one requiring only the necessary resources, and doing so optimally. 

As illustrated in Section \ref{sec:exp}, our implementation is capable of synthesizing optimal implementations for error sizes down to $10^{-15}$.  To calculate how small of an error one might need to approximate a single-qubit unitary to, consider Shor's integer factoring algorithm.  Suppose we want to factor a 1,000,000-bit number, and the effect of the error due to gate approximation is required to be negligible, e.g., 0.01\%.  The number of single-qubit gates requiring approximation is about $2n\log{n}|_{n=1,000,000}\approx 4{\cdot}10^7$.  Assuming the errors add up, the precision of each individual gate does not need to be smaller than $10^{-12} < \frac{0.0001}{4{\cdot}10^7}$.  As such, since our algorithm capable of approximating single-qubit unitaries to error $10^{-15}$, it can be readily used to approximate the QFT (the modular exponentiation can be implemented exactly and requires no approximations) part of Shor's algorithm that factors 1,000,000-bit numbers.

The above simple rough calculation motivated our decision to invest additional resources into the calculations in exchange for a higher quality output (which lead to a lower cost quantum circuit).  In particular, we noted that we can manage precisions of practical importance, and have thus invested the additional time into computing the best approximating unitary---our results are accompanied by the optimality guarantee.  Furthermore, in the above calculation we assumed that the errors add up in the worst possible way.  Naturally and for most applications and approximations this is unlikely to be the case, as random and independent noise scales as the square root of the sum of absolute values of all errors.  Our algorithm may furthermore be easily updated to provide slightly (on the order of 5\%) suboptimal implementations and draw a random one, thereby providing a way to grow logical error by a sublinear function of the sum of errors and potentially resulting in significant advantage via savings 
in approximating with a much larger error.

\section{Preliminaries}
In this section we review basic concepts with the goal of introducing the notations.  For an in-depth review, please see \cite{bk:nc}. 

While the state of a classical bit can be either $0$ or $1$, the state of a quantum bit, or {\it qubit}, is described by a unit vector in the two-dimensional complex vector space $\mathbb{C}^2$.  It is common to use the notation $\ket{0}$ and $\ket{1}$ to denote an orthonormal basis of the state space and refer to it as the computational basis.  The state of the system of $n$ qubits is described by a unit vector that belongs to the $n$-fold tensor product of the two-dimensional complex spaces.  If the state of the first qubit is $\alpha \ket{0}+ \beta \ket{1}$ and the state of the second qubit is 
$\ket{0}$, then the state of the corresponding two-qubit system is 
\[
 \left( \alpha \ket{0}+ \beta \ket{1} \right) \otimes \ket{0} \in \mathbb{C}^2\otimes\mathbb{C}^2 \cong \mathbb{C}^4. 
\]
In the circuit model, quantum computation is performed by applying unitary operators (linear operators that preserve the usual inner product between vectors) to the state vector at discrete time steps.  A unitary operator $U$ applied to a single qubit corresponds to the tensor product $U\otimes I$ 
with identity on the rest of the qubits.  For example, when one applies $U$ 
to the first qubit and the system is in the state mentioned above, the result becomes 
\[
  U \otimes I \left( \left( \alpha \ket{0}+ \beta \ket{1} \right) \otimes \ket{0} \right) =  \left( \alpha U \ket{0}+ \beta U \ket{1} \right) \otimes \ket{0}.
\]

\begin{figure}
\raisebox{-0.4\height}{\includegraphics[scale=0.4]{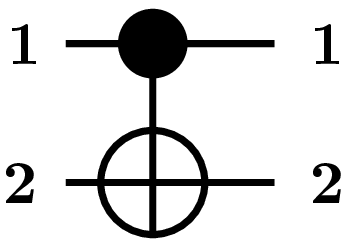}}
\hfill{}%
$\left(\begin{array}{cccc}
1 & 0 & 0 & 0\\
0 & 1 & 0 & 0\\
0 & 0 & 0 & 1\\
0 & 0 & 1 & 0
\end{array}\right)$
\hfill{}%
$\begin{array}{c}
\ket{00}\rightarrow\ket{00} \\
\ket{01}\rightarrow\ket{01} \\
\ket{10}\rightarrow\ket{11} \\
\ket{11}\rightarrow\ket{10}
\end{array}$

\caption{\label{fig:CNOT}CNOT gate with control on the first qubit and target on the second qubit. From
left to right: diagram of the gate, its corresponding unitary, action
on the computational basis of the two qubit state space.}
\end{figure}

Depending on the gate library used, it is not always possible to implement 
a unitary operator exactly, and an approximation is used instead, meaning those unitaries that cannot be implemented exactly are replaced with the ones implementable and close in some distance defined on the unitaries. 

The gate library we focus on is Clifford+T. It consists of the following single-qubit gates
\[
 X:= \l(\begin{array}{cc} 0 & 1 \\ 1 & 0 \\ \end{array}\r), Z:= \l(\begin{array}{cc} 1 & 0 \\ 0 & -1 \\ \end{array}\r), Y:= iXZ
\]
\[
 H:= \frac{1}{\sqrt{2}}\l(\begin{array}{cc} 1 & 1 \\ 1 & -1 \\ \end{array}\r),
  T:= \l(\begin{array}{cc} 1 & 0 \\ 0 & e^{i\pi/4} \\ \end{array}\r), P:=T^2
\]
and the CNOT gate~(see Fig.~\ref{fig:CNOT}), which is a two-qubit gate.  The above gates, except T gate, are all Clifford gates and typically are easier to implement fault-tolerantly than the T gate.  The H gate is also called the Hadamard gate and the P gate is also called the Phase gate; X, Y, and Z are known as single-qubit Pauli gates. 
 
The result of the computation is obtained by the measurement of the resulting state.  For a single-qubit state $\alpha\ket{0}+\beta\ket{1}$, the 
probability of the outcome $0$ is $|\alpha|^2$, and the probability of the outcome $1$ is $|\beta|^2$.  In general, the result of an $n$-qubit computation 
is a probability distribution on all Boolean $n$-bit strings.  The precision that we want to achieve during the approximation procedure is determined 
by the precision of the resulting distribution of the outcomes required to successfully obtain the correct answer.



Any single-qubit unitary can be decomposed in terms two Hadamard gates and Z-rotations
\[
\Rzp := \left(\begin{array}{cc}
e^{-i\p/2} & 0\\
0 & e^{i\p/2}
\end{array}\right)
\]
(see, for example, \cite{bk:ksv}, solution to Problem 8.1).  Therefore, the ability to approximate $\Rzp$ implies the ability to approximate any single-qubit unitary.  In addition, $\Rzp$ are common single-qubit rotations used in many quantum algorithms.  For example, they are used in the Quantum Fourier Transform---an important ingredient in a number of quantum algorithms, as well as in quantum chemistry simulations.  In this paper we focus on optimal approximation of $\Rzp$ by Clifford and T gates, and thereby address the practical needs. 


\section{Main results}
In this section we describe our main result---the algorithm for approximating single-qubit rotations $\Rzp$ using Clifford and T circuits. The algorithm is based on exact synthesis results \cite{arXiv:1206.5236}. In particular, it was previously shown~\cite{arXiv:1206.5236} that any single-qubit unitary can be represented by a Clifford and T circuit if and only if it has the following form: 
\[
 U[x,y,k] = \left( \begin{array}{cc}
x & -y^{\ast}\w^k \\
y & x^{\ast}\w^k \\
\end{array} \right),  
\]
where $x, y \in \Zf$ and $\w := e^{i\pi/4}$. We call these unitaries \emph{exact}, implying that the unitaries can be \emph{exactly} represented by the respective circuits, and do not require to rely on approximation.  To approximate $\Rzp$ we first find a unitary $U[x,y,k]$ that is close to it and then use the exact-synthesis algorithm from~\cite{arXiv:1206.5236} to find a circuit implementing $U[x,y,k]$ with the optimal number of H and T gates.  With the above exact synthesis results~\cite{arXiv:1206.5236} we reduce the problem of finding the best approximation by a Clifford and T circuit to the problem of finding the best approximation by an exact unitary. We call this problem the Closest Unitaries Problem ($\mathrm{CUP}$). 

\subsection{Closest Unitaries Problem}
Here we define the Closest Unitaries Problem formally and briefly discuss why it is easier to solve this problem than the similar problem involving circuits.  We use global phase invariant distance to measure the quality of approximation. It is defined on single-qubit unitaries as
\[
d(U,V) := \sqrt{1-\left|tr(UV^\dg)\right|/2}. 
\]
Motivated by the relative difficulty of implementing the T gate in practice, we aim to find the best approximation using at most the given number of T gates. We next introduce the T-count, $\T(U)$, to be the minimal number of T gates required to implement $U$ up to the global phase as a circuit over the Clifford and T gate library.  In other words, $U$ can be written in the form 
\[
 e^{i\alpha} C_1 T C_2 T \ldots C_n T C_{n+1},
\]
where $e^{i\alpha}$ is some constant (global phase), $C_i$ are Clifford unitaries, $n=\T(U)$, and $U$ cannot be written in the above form for any $n<\T(U)$. 

\begin{prob}
$\mathrm{CUP}[n,\p]$ (Closest Unitaries Problem) is the problem of finding:
\begin{itemize}
 \item the distance $\ve[n,\p]$ between $\Rzp$ and the set of exact unitaries with T-count at most $n$,
 \item  the subset $D[n,\p]$ of all exact unitaries with T-count at most $n$ and within distance $\ve[n,\p]$ from $\Rzp$; T-count of all elements of $D[n,\p]$ must be minimal. 
\end{itemize}
\end{prob}
The requirement for the elements of $D[n,\p]$ to have the minimal T-count is non-trivial.  This is because the set of all exact unitaries with T-count at most $n$ and within distance $\ve[n,\p]$ from $\Rzp$ may contain unitaries with different T-count, as is illustrated in the example we give in the next subsection (formulas (\ref{ex:exampl})). 

A na\"{\i}ve brute-force solution to $\mathrm{CUP}[n,\p]$~\cite{F1} requires one to enumerate all exact unitaries with T-count at most $n$. Our algorithm allows us to significantly reduce the size of the search space used for solving $\mathrm{CUP}[n,\p]$ when $\ve[n{-}1,\p]$ is known. Informally, if one has already solved  $\mathrm{CUP}[n{-}1,\p]$ there is no need to solve  $\mathrm{CUP}[n,\p]$ from scratch: one just needs to check if using exact unitaries with T-count at most $n$ instead of exact unitaries with T-count at most $n{-}1$ allows them to improve the quality of approximation over previously achieved $\ve[n{-}1,\p]$. It is much easier to accomplish this when approximating by unitaries compared to approximating by circuits.  

In the next section we describe in more detail the problem that we need to solve on top of $\mathrm{CUP}[n{-}1,\p]$ to find the solution to $\mathrm{CUP}[n,\p]$. We call this problem the Restricted Closest Unitaries Problem. 

\subsection{Restricted Closest Unitaries Problem}
We introduce the notion of minimal unitaries that is crucial for the definition of the Restricted Closest Unitaries Problem~($\mathrm{RCUP}$) and use it to show the relation between $\mathrm{CUP}$ and $\mathrm{RCUP}$. The definition of the minimal unitaries is motivated by the fact that the distance between the exact unitary $U[x,y,k]$ and $\Rzp$ can be written as 
\begin{equation}\label{eq:dist}
  d(\Rzp,U[x,y,k]) = \sqrt{1-\l|\re(xe^{i\p/2}\w^{-k/2})\r|}.
\end{equation}
In particular, we see that the distance depends only on $x$ but not on $y$. We say that unitary $U[x,y,k]$ is \emph{minimal} if its T-count is equal to the minimum of T-counts over all unitaries of the form $U[x,y',k]$ for $y'$ from $\Zf$. Below is an example of the minimal and non-minimal unitaries that we found while approximating $R_z(\pi/16)$: 
\begin{equation}\label{ex:exampl}
\begin{array}{l}
 U[3{+}5\w{-}3\w^2{-}2\w^4,-2{+}2\w^2{-}3\w^3,0]/8 \text{and} \\
 U[3{+}5\w{-}3\w^2{-}2\w^4, 3{-}2\w{+}2\w^3,0]/8.
\end{array}
\end{equation}
The first unitary has T-count of 10 and the second has the T-count of 12.  

We next state the RCUP using the notion of minimal unitaries: 
\begin{prob}
$\mathrm{RCUP}[n,\p,\d]$ (Restricted Closest Unitaries Problem) is the problem of finding: 
\begin{itemize}
 \item The distance $\ve[n,\p,\d]$ between $\Rzp$ and the set of minimal exact unitaries within distance $\d$ from $\Rzp$ and with T-count equal to $n$ (in the case if there are no such unitaries we define $\ve[n,\p,\d]:=\d$),
 \item The set $D[n,\p,\d]$ of minimal exact unitaries within distance $\ve[n,\p,\d]$ from $\Rzp$ and with the T-count equal to $n$. 
\end{itemize}
\end{prob}

The following Lemma establishes the relation between CUP and RCUP.
\begin{lem}\label{lem:cup}
$\mathrm{CUP}[n,\p]$ reduces to $\mathrm{CUP}[n{-}1,\p]$ and $\mathrm{RCUP}[n,\p,\d]$ for $\d=\ve[n{-}1,\p]$ as follows:
\begin{itemize}
 \item If $\ve[n,\p,\d] \ge \ve[n{-}1,\p]$, then $\ve[n,\p]=\ve[n{-}1,\p]$ \\ and $D[n,\p]=D[n{-}1,\p]$.
 \item If $\ve[n,\p,\d] < \ve[n{-}1,\p]$, then $\ve[n,\p]=\ve[n,\p,\d]$\\ and $D[n,\p]=D[n,\p,\d]$.
\end{itemize}
\end{lem}
\begin{proof}
There are two alternatives for a given pair $n$ and $\p$: using unitaries with T-count $n$ in addition to unitaries with T-count $n{-}1$ either allows one to achieve better approximation quality or it does not.  In the first case, the only possibility is if $\ve[n,\p,\d] < \ve[n{-}1,\p]$; in the second case, if $\ve[n,\p,\d] \ge \ve[n{-}1,\p]$.  By the definition $D[n,\p]$ contains only minimal unitaries, and the condition $\ve[n,\p] < \ve[n{-}1,\p]$ implies that all unitaries in $D[n,\p]$ must have T-count equal to $n$. Therefore, we conclude that $D[n,\p,\d]=D[n,\p]$ when $\ve[n,\p,\d] < \ve[n{-}1,\p]$. It is also easy to see that if using unitaries with T-count $n$ does not improve the approximation quality, then $D[n,\p]=D[n{-}1,\p]$. 
\end{proof}
In practice, we believe it to be helpful to have a list of answers to the set of problems $\mathrm{CUP}[n,\p]$ for $n$ between $0$ and $N$, where $N$ is bounded by the amount of classical computing resources available, but has to be large enough to allow scalable quantum computing.  Indeed, such a list allows the compiler (or a circuit designer) to quickly select the best needed approximation for a given unitary.  The list is furthermore not very long owing to the logarithmic scaling of the optimal $T$-count as a function of the approximation error, allowing efficient storage and access to it.  Lemma~\ref{lem:cup} shows that the task of computing such a list is equivalent to solving a set of $\mathrm{RCUP}[n,\p,\d_n]$, with a proper choice of $\d_n$. 

\subsection{Algorithm}
In this section we present an algorithm for solving $\mathrm{RCUP}[n,\p,\d]$ and prove its correctness. In the previous section we showed that the distance between $\Rzp$ and the exact unitary $U[x,y,k]$ is a function of $x$ and $k$, see equation (\ref{eq:dist}). Our algorithm searches for approximations of $x$ instead of directly searching for $U[x,y,k]$. This motivates the following definition of T-count, applied to the elements of $\Zf$: 
\[
 \T_k(x):= \min \left\{ \T(U[x,y,k]) | U[x,y,k] \text{ -- exact unitary} \right\} 
\]
If the above minimum is to be taken over an empty set, we define $\T_k(x):=\infty$. In other words, this means that there is no unitary over the ring $\Zf$ such that $x$ is its entry. We discuss the conditions on $x$ that guarantee the existence of $U[x,y,k]$ later in Section \ref{sec:norm-eq}. Function $\T_k(x)$  is useful for both the algorithm description and the proof of its correctness. It has several properties:
\begin{prop}\label{prop:t-count}
The T-count of an element $x \in \Zf$ has the following properties:
\begin{itemize}
 \item T-count of any minimal unitary $U[x,y,k]$ is $\T_k(x)$,
 \item $\T_k(x)=\T_{k\md2}(x)$,
 \item $\T_k(x)=\T_k(\w x)$,
 \item if $4\le\T_k(x) < \infty$  \\ then $\T_k(x){=}\sde(|x|^2){-}2{+}(\sde(|x|^2){+}k)\mathrm{mod\,}2$.
\end{itemize}
\end{prop}

The function $\sde(x)$ (the smallest denominator exponent) is defined to take value $m$ when $x$ is written in the form 
$(a+\sqrt{2}b)/\sqrt{2}^m$, where $a, b$ and $m$ are integers.  

To solve $\mathrm{RCUP}[n,\p,\d]$ we go through the elements $x$ of $\Zf$ such that $\T_k(x)=n$ in an efficient way.  In particular, we split the search problem into two smaller sub-problems---the search for real and imaginary parts of $x$, and take into account the necessary conditions that the real and imaginary parts of the solution must satisfy.  The latter helps to shrink the size of the search space, resulting in a better efficiency.  In more detail, the properties of $\T_k(x)$ imply that we only need to consider $k$ equal to zero or one, and limit the set of possible $x$ using the relation between $\T_k(x)$ and $\sde$.  The next proposition summarizes the constraints on $x$ that must be satisfied. 

\begin{prop}\label{pr:denom}
Let $x \in \Zf$ be such that $\T_k(x)\ge4$.  Then, $x$ can be written as $(a_0 + \sqrt{2}b_0 + i(a_1+\sqrt{2}b_1)) / \sqrt{2^m}$ for integers $a_j,b_j$ and $m$.  The following conditions hold: 
\begin{itemize}
 \item $m \le \lfloor(\T_k(x)+5)/2\rfloor$,
 \item $a_0^2+2b_0^2 \le 2^m \text{ and } a_1^2+2b_1^2 \le 2^m$.
\end{itemize}
\end{prop}

Note that we separated conditions on integers $a_0,b_0$ and $a_1,b_1$ defining the real and imaginary parts of $x$. The first constraint follows from the inequality $\sqrt{1-\l|\re(xe^{i\p/2}\w^{-k/2})\r|} \le \d$ (see (\ref{eq:dist})) and leads to additional constraints on $a_j,b_j$ that are also separate for the real and imaginary parts of $x$: 

\begin{prop}\label{pr:delta}
Let $x = (a_0 + \sqrt{2}b_0 + i(a_1+\sqrt{2}b_1)) / \sqrt{2^m}$ for integers $a_j,b_j$ and a non-negative integer $m$, and $|x|^2\le1$. Let $\sqrt{1-\left|\re(xe^{-i\theta})\right|} \le \d \le 1$.  Then, the following conditions hold:  
\begin{itemize}
 \item $|(a_0+\sqrt{2}b_0)-\cos(\theta)\sqrt{2^m}|\le \d\sqrt{2^{m+1}}$,
 \item $|(a_1+\sqrt{2}b_1)-\sin(\theta)\sqrt{2^m}|\le \d\sqrt{2^{m+1}}$.
\end{itemize}
\end{prop}

In the first part of the algorithm (Fig.~\ref{fig:rcu-algorithm}) we build arrays of $a_j,b_j$ satisfying conditions from Propositions~\ref{pr:denom} and \ref{pr:delta} by calling the FIND-HALVES~(Fig.~\ref{alg:find-halves}) procedure. In addition, in FIND-HALVES we compute contributions $\ve_{re}$ and $\ve_{im}$ from the real and imaginary parts of $x$ to $\sqrt{1-\l|\re(xe^{i\p/2}\w^{-k/2})\r|}$. More details on this are provided by the following proposition:

\begin{prop}\label{pr:split}
Let $\d\le1/2$ and let $x = (a_0 + \sqrt{2}b_0 + i(a_1+\sqrt{2}b_1)) / \sqrt{2^m}$ for integers $a_j,b_j$ and non-negative integer $m$. If the following holds:
\begin{itemize}
 \item $|(a_0+\sqrt{2}b_0)-\cos(\theta)\sqrt{2^m}|\le \d\sqrt{2^{m+1}}$,
 \item $|(a_1+\sqrt{2}b_1)-\sin(\theta)\sqrt{2^m}|\le \d\sqrt{2^{m+1}}$,
\end{itemize}
then
\[
\begin{array}{l}
(\sqrt{1-\left|\re(xe^{-i\theta})\right|})^2 \sqrt{2^{m}}=\ve_{re}+\ve_{im}, \\
\quad\ve_{re}=\cos(\theta)(\sqrt{2^m}\cos(\theta)-(a_0+b_0\st)),\\%
\quad\ve_{im}=\sin(\theta)(\sqrt{2^m}\sin(\theta)-(a_1+b_1\st)).
\end{array}
\]
\end{prop}
In the next steps of the algorithm we enumerate $x$ from $\Zf$ satisfying the necessary conditions. We start with those that give the best approximations, in other words the smallest value of $\sqrt{1-\l|\re(xe^{i\p/2}\w^{-k/2})\r|}$. For each candidate $x$ we compute $\T_k(x)$ using procedure MIN-T-COUNT. When an $x$ with the required T-count is found, procedure ALL-UNITARIES is called to find all minimal exact unitaries of the form $U[x,y,k]$, and the algorithm terminates.  Details on MIN-T-COUNT and ALL-UNITARIES are provided in Section~\ref{sec:norm-eq}. 

It is important to note that step \ref{line:merge} of the algorithm (Fig.~\ref{fig:rcu-algorithm}) is performed efficiently. We first choose tuples corresponding to the real parts such that $\ve_{re}$ belongs to the interval $I=[\alpha_0,\alpha_1]$ and then, for each $\ve_{re}$, choose tuples with $\ve_{im}$ in the interval $[\alpha_0{-}\ve_{re},\alpha_1{-}\ve_{re}]$. The proofs of the propositions presented above are tedious and we postpone them to Section~\ref{sec:td}. We next rely on these propositions to prove the correctness of the algorithm. 

\begin{figure}[t]
\begin{algorithmic}[1]
\Require $n,\p,\d$ \Comment{$n$ -- T-count, $\Rzp$ -- target rotation}
\State $m \gets \lf (n+1)/2 \rf+2$
\For{$k=0,1$}
  \State $L_{re,k} \gets \text{FIND-HALVES}(\cos(\p-\pi k/8),m,\d)$
  \State $L_{im,k} \gets \text{FIND-HALVES}(\sin(\p-\pi k/8),m,\d)$
  \Statex[4] \Comment (described on Fig. \ref{alg:find-halves})
\EndFor 
\State Interval $I \gets [0,\alpha]$ \Comment Pick $\alpha > 0$ based on $L_{re,k},L_{im,k}$
\While{$I \cap [0,\d] \ne \varnothing$}
\State {Find an array $A$ of tuples $(\ve,a_0,b_0,a_1,b_1,k)$ s.t.: \label{line:merge}
\Statex[2] $ \bullet\, (\ve_{re},a_0,b_0)$ from $L_{re,k}$ 
\Statex[2] $ \bullet\, (\ve_{im},a_1,b_1)$ from $L_{im,k}$}
\Statex[2] $ \bullet\, \ve=\ve_{re}+\ve_{im}$ and $\ve \in I \cap [0,\d]$
\State Sort $A$ by $\ve$ in ascending order\
\State {$\ve_1 < \ldots < \ve_M \gets$ all distinct $\ve$ that occur in $A$}
\For{$ j = 1 \text{ to } M $}
  \State $\sol \gets \varnothing$
  \ForAll{$(\ve_j,a_0,b_0,a_1,b_1) \in A$}
    \State $x' \gets a_0+b_0\sqrt{2}+i(a_1+b_1\sqrt{2})$
    \State $n_0 \gets \text{MIN-T-COUNT}(x',m,k)$
    \Statex[4] \Comment (computes $\T_k(x'/\sqrt{2^m})$, see Sec. \ref{sec:norm-eq})
    \If{$n=n_0$}
      \State $\sol \gets \sol \cup \text{ALL-UNITARIES}(x',m,k)$
      \Statex[5] \Comment (enumerates minimal unitaries
      \Statex[5] \Comment $U[x'/\sqrt{2^m},y,k]$, see Sec. \ref{sec:norm-eq})
    \EndIf
  \EndFor
  \If{$\sol \ne \varnothing$}
  \State \Return $(\ve_j,\sol)$ \Comment Solution
  \EndIf
\EndFor
\State Replace $I=[\alpha_0,\alpha_1]$ by $I=[\alpha_1,2\alpha_1-\alpha_0]$ 
\EndWhile
\State \Return $(\d,\varnothing)$ \Comment No solutions
\Ensure $(\ve_n^R,\sol_{n,\p}^\d )$

\end{algorithmic}
\caption{\label{fig:rcu-algorithm}RCU-Algorithm: the algorithm for  $\mathrm{RCUP}[n,\p,\d]$.}
\end{figure}

\begin{thm}\label{thm:rcu-algorithm}
The RCU-Algorithm (Fig.~\ref{fig:rcu-algorithm}) solves $\mathrm{RCUP}[n,\p,\d]$---the Restricted Closest Unitaries Problem with T-count $n$, angle $\p$ and threshold $\d$ when $n\ge4$ and $\d \le 1/2$. 
\end{thm}
\begin{proof}
We first formally describe the output of the algorithm $(\vea,\sol)$ and then prove that it is indeed the solution to $\mathrm{RCUP}[n,\p,\d]$. Let us define $\th_k:=\frac{\pi k}{8}-\frac{\p}{2}$ for $k=0,1$ and the following sets: 
\[
 S_k:=\sets{x}{\T_k(x)=n,\,\sqrt{1-\l|\re(xe^{-i\theta_k})\r|} \le \d}.
\]
We first consider the case when at least one of the $S_k$ is non-empty and show that the algorithm outputs a pair $(\vea,\sol)$ such that 
\[
\begin{array}{l}
 \vea = \min(\vea_1,\vea_2),\text{ where } \\
 \vea_k = \min\set{\sqrt{1-\l|\re\l( x e^{-i\theta_k}\r)\r|}}{ x  \in S_k }.
\end{array}
\]
Let us also denote by $\sol_k$ the elements of $S_k$ within distance $\vea$ from $\Rzp$. It is not difficult to see that at least one of the $\sol_k$ is non-empty. 

By the definition of $S_k$, the value $\vea$ is in the interval $[0,\d]$.  Therefore, at some iteration in the \emph{while} loop $\vea$ will belong to the interval $I$ and will be in the list $\ve_1,\ldots,\ve_M$. Indeed, suppose that $\sol_{\ks}$ is non-empty and $x$ is its element. In other words $x$ is such that $\T_{\ks}(x)=n$ and $\sqrt{1-\l|\re(xe^{-i\theta_{\ks}})\r|} = \vea$. Proposition~\ref{pr:denom} implies that $x$ can be represented as 
\[
  (a_0+b_0\st+a_1i  +b_1 i\st)/\sqrt{2^m}\text{, for } m = \lf (n+1)/2 \rf+2.
\]
Propositions \ref{pr:denom} and \ref{pr:delta} imply that integers $a_j,b_j$ must satisfy the following inequalities
\[
 \begin{array}{l}
  a_0^2+2b_0^2 \le 2^m, a_1^2+2b_1^2 \le 2^m, \\
  |(a_0+\sqrt{2}b_0)-\cos(\theta_{\ks})\sqrt{2^m}|\le\sqrt{2^{m+1}}\d, \\
  |(a_1+\sqrt{2}b_1)-\sin(\theta_{\ks})\sqrt{2^m}|\le\sqrt{2^{m+1}}\d.
 \end{array}
\]
This implies that after executions of procedure FIND-HALVES, for
\[
\begin{array}{r}
 \ve_{re}=\sqrt{2^{-m}}\cos(\theta)(\sqrt{2^m}\cos(\theta)-(a_0+b_0\st)) \\
 \ve_{im}=\sqrt{2^{-m}}\sin(\theta)(\sqrt{2^m}\sin(\theta)-(a_1+b_1\st))
\end{array},
\]
the triples $(\ve_{re},a_0,b_0)$ and $(\ve_{im},a_1,b_1)$ belong to $L_{re,\ks}$ and $L_{im,\ks}$ (Fig.~\ref{fig:rcu-algorithm}), correspondingly.  From Proposition~\ref{pr:split} we recall that $\vea=\ve_{re}+\ve_{im}$, and therefore the tuple $(\vea,a_0,b_0,a_1,b_1)$ belongs to the array $A$ and $\vea$ is in the list $\ve_1,\ldots,\ve_M$.  Let $m_0$ denote the position of $\vea$ in the list. When the \emph{for} loop reaches $m=m_0$ the algorithm will terminate. It is not difficult to see that the algorithm does not terminate before this step, as it would contradict the minimality of $\vea$. The only way for the the algorithm to terminate earlier is if there is an $x$ such that $\sqrt{1-\l|\re\l(x_k e^{-i\theta_k}\r)\r|}<\vea$ and $\T_k(x)=n$. 

The procedure ALL-UNITARIES is designed to output the set
\[
 \sol = \bigcup_{\substack{k=0,1 \\ g=0..7}} \set{ \w^g U[x,y,k] \in \Um }{ x \in \sol_k, y \in \Zf },
\]
where by $\Um$ we denote the set of all exact minimal unitaries. 

Let us now show that $(\vea,\sol)$ is the solution to $\mathrm{RCUP}[n,\p,\d]$. 
Suppose that the set $D[n,\p,\d]$ is non-empty and let $U[x,y,k]$ be some element of the set such that the distance to it from $\Rzp$ is minimal. The distance can be expressed as: 
\[
 d(\Rzp,U[x,y,k])=\sqrt{1-\l|\re(xe^{i\p/2}\w^{-k/2})\r|} < \d.
\]
Proposition~\ref{prop:t-count} implies that $\T_k(x)=n$ because $U[x,y,k]$ is a minimal unitary with T-count $n$. We next show that we can make $k$ equal to zero or one. Indeed, let us write $k=k_0+2s$ where $k_0$ is either zero or one, then 
\[
\l|\re(xe^{i\p/2}\w^{-k/2})\r|= \l|\re(x\w^s e^{i\p/2}\w^{-k_0/2})\r|.
\]
Again using Proposition~\ref{prop:t-count} we see that $\T_k(x\w^s)=n$. In addition, we have 
\[
 d(\Rzp,U[x,y,k])=\sqrt{1-\l|\re((x\w^s)e^{-i\th_{k_0}})\r|}.
\]
This implies that $x\w^s$ is in $S_{k_0}$ and $\vea$ is less than or equal to $\ve[n,\p,\d]$. We next show that $\vea\ge\ve[n,\p,\d]$. Suppose that $\sol_{k_0}$ is non-empty for some $k_0$ and let $x$ be its element. Then, we have $\T_{k_0}(x)=n$ and there exists a $y \in \Zf$ such that unitary $U[x,y,k_0]$ is minimal with T-count $n$. We also notice that $d(\Rzp,U[x,y,k_0])=\vea$, which concludes the proof of the equality $\vea=\ve[n,\p,\d]$. 

Next we show that $\sol$ coincides with the set $D[n,\p,\d]$. Consider some element $\w^g U[x,y,k]$ of $\sol$. The set $D[n,\p,\d]$ contains any unitary $U$ together with all unitaries $\w^g U$, therefore it is enough to show that $U[x,y,k]$ is in $D[n,\p,\d]$. The fact that $U[x,y,k]$ is a minimal unitary and $\T_k(x)=n$ implies that $U[x,y,k]$ has T-count $n$. On the other hand, by definition of $\sol_k$ we have that $d(\Rzp,U[x,y,k])=\vea$ which shows that $\sol$ is a subset of $D[n,\p,\d]$. Let us now show that $D[n,\p,\d]$ is a subset of $\sol$. Let $U[x,y,k]$ be an element of $D[n,\p,\d]$. We first note that $U[x,y,k]$ can be equivalently written as $\w^s U[x\w^{-s},y\w^{-s},k{-}2s]$. We chose $s$ in such a way that $k_0=k{-}2s$ is either zero or one. We note that $U[x\w^{-s},y\w^{-s},k_0]$ is a minimal unitary and therefore $\T_{k_0}(x\w^{-s})=n$. Distance $d$ is global phase invariant, which implies that $d(\Rzp,U[x\w^{-s},y\w^{-s},k_0])=\vea$ and $\sqrt{1-\l|\re\l( x\w^{-s} e^{-i\theta_k}\r)\r|} = \vea$. 
We 
conclude that $x\w^{-s}$ is in $\sol_{k_0}$ and $U[x\w^{-s},y\w^{-s},k-2s]$ 
is in $\sol$. It is not difficult to see from the definition of $\sol$ that if $U$ is in $\sol$ then for any integer $g$ unitary $\w^g U$ is also in $\sol$. This concludes the proof of the equality $D[n,\p,\d]=\sol$. 

To conclude the entire proof, we still need to discuss the special case when the problem has no solutions. Suppose that $D[n,\p,\d]$ is an empty set. It is not difficult to show that this implies that both $S_k$ are empty and vice-versa, using the ideas from the main part of the proof. 
\end{proof}

The restrictions on $n$ and $\d$ in the theorem statement are not significant. It is much easier to solve $\mathrm{CUP}[n,\p]$ directly when $n<4$. From our numeric experiments we found that $D[3,\p]$ is always less than $0.1376$ therefore each time $\mathrm{RCUP}[n,\p,\d]$ is used, it is used with parameter $\d<1/2$. 

\begin{figure}[t]
\begin{algorithmic}[1]
\Require $ \alpha,\d \in \R, m \in \Z, m \ge 0$
\Procedure{$\text{FIND-HALVES}$}{$\alpha,m,\d$}
  \State $W \gets \alpha\sqrt{2^{-m}}, \ve \gets \d\sqrt{2^m}$
  \State $b \gets \lf-\sqrt{2^{m}}\rf $
  \State $v \gets \alpha\sqrt{2^{m}} - b\st $ \Comment true on every step
  \State $R \gets \varnothing $ 
  \While{$b \le \lceil \sqrt{2^m} \rceil $}  
    \State $a_{\min}= \lceil v - \ve \rceil, a_{\max}=\lfloor v+\ve \rfloor$
    \ForAll{$a \in [a_{\min},a_{\max}] \cap \Z$} \Comment See Sec. [*]
    \If{$a^2+2b^2\le2^m$} \Comment for discussion
      \State $R\gets R \cup \{((x-a)W,a,b)\}$ 
    \EndIf
    \EndFor 
    \State $b \gets b + 1, x \gets x - \sqrt{2}$
  \EndWhile
  \State Sort $R$ by first element in ascending order 
  \State \Return R
\EndProcedure
\end{algorithmic}
\caption{\label{alg:find-halves}FIND-HALVES Procedure. Finds all numbers of the from $a+\sqrt{b}$ that satisfy conditions $|\alpha\sqrt{2^m} - (a+ b\st)|\le\d\sqrt{2^m}$ and $a^2+2b^2 \leq 2^m$.  Returns the list of tuples $(\alpha\sqrt{2^{-m}}(\alpha\sqrt{2^m}  - (a+ b\st)),a,b)$ sorted by the first entry.}
\end{figure}

It is possible to make the FIND-HALVES procedure~(Fig.~\ref{alg:find-halves}) slightly more efficient. We found that the length of the interval $[a_{\min},a_{\max}]$ in it is usually less than $1/2$ therefore the internal \emph{for} loop can be replaced with the function \emph{round}. It is also not difficult to see that the \emph{while} loop of the procedure can be easily parallelized. In our implementation of the algorithm we benefit from both these observations.




\section{Technical details}\label{sec:td}

In this section we prove Propositions \ref{prop:t-count}-\ref{pr:split}. First we need to recall some useful results and definitions from~\cite{arXiv:1206.5236}. It is possible to extend $\sde$ on exact unitaries as $\sde(U[x,y,k])=\sde(|x|^2)$. The following result relates it to a unitary T-count: 
\begin{lem}[Corollary of the Theorem 2 from \cite{arXiv:1206.5236}]
\label{thm:tcount}Let $U$ be a unitary over $\Zf$ such that $\sde(U)\ge4$, and $j$ and $l$ be integers such that $\sde(HT^{j}UT^{l}H) = \sde(U) + 1$. Then, 
\[
\T(U)=\sde(U) - (j\md2) - (l\md2).
\]
\end{lem}
Next follow the proofs of all four propositions. 

\begin{proof}[Proof of Proposition~\ref{prop:t-count}]
The first property follows directly from the definition of $\T_k(x)$ and minimal unitaries.  To prove the second one we observe that multiplication by the Phase gate $P:=diag\{1,i\}$ does not change the T-count of a unitary and $U[x,y,k]P=U[x,y,k+2]$; the definition of $\T_k(x)$ implies $\T_k(x)=\T_{k+2}(x)$.  To prove the third one we rely on the equality $\w U[x,y,k] = U[\w x,\w y,k+2]$. 

To prove the fourth property let us consider the minimal unitary $U[x,y,k]$.  Its T-count is at least four and therefore it requires at least three Hadamard gates to be implemented.  As a result, $\sde(U[x,y,k])$ is greater than four. Lemma~\ref{thm:tcount} applies to $U[x,y,k]$ and implies that there are three possible values of $\T_k(x)$---being $\sde(U[x,y,k])-2$, $\sde(U[x,y,k])-1$, and $\sde(U[x,y,k])$.  A minimal unitary cannot have T-count equal to $\sde(U[x,y,k])$.  Indeed, if this were true the unitary $TU[x,y,k]T^{\dagger}$, that is equal to $U[x,y\w,k]$, would have had T-count of $\sde(U[x,y,k])-2$, which contradicts the minimality of $U[x,y,k]$.  We next show that the T-count of $U[x,y,k]$ is completely defined by the parity of $\sde$ and $k$.  The determinant of $U[x,y,k]$ is equal to $\w^k$.  In addition, the T-count of $U[x,y,k]$ must have the same parity as $k$, because the T gate is the only gate in a Clifford+T library whose determinant is equal to an odd power of $\w$. To illustrate, if $\sde$ were odd and k were even, the T-count could only be equal to $\sde(U[x,y,k])-1$. Via going through all possible parity combinations, we get to the required expression for $\T_k(x)$.
\end{proof}

\begin{proof}[Proof of Proposition \ref{pr:denom}]
Recall that any $x \in \Zf$ can be written as $(a_0 + \sqrt{2}b_0 + i(a_1+\sqrt{2}b_1)) / \sqrt{2^m}$, where $a_j$, $b_j$ and $m$ are integers.  Let us choose such a representation where $m$ is minimal.  Integer $m$ must be positive, otherwise $x$ either cannot be an entry of a unitary or its T-count must be zero.  Note that at least one of the $a_j$ in the expression must be odd, otherwise $m$ is not minimal. It is useful to expand $|x|^2$ as
\[
\frac{(a_0^2 + a_1^2) + 2(b_0^2 + b_1^2) + 2\sqrt{2}(a_0 b_0 + a_1 b_1) }{2^m}.
\]
If one of the $a_j$ is odd and the other one is even we obtain the equality $\sde(|x|^2) = 2m$. Let us not consider the case when both $a_j$ are odd.  In the case that $b_0$ and $b_1$ have different parity we get:
\[
\begin{array}{r}
 (a_0^2 + a_1^2) + 2(b_0^2 + b_1^2) = 0 \md 4, \text{ and } \\
 a_0 b_0 + a_1 b_1 = 1 \md 2,
 \end{array}
\]
and conclude that $\sde(|x|^2)=2m-3$. In the other case, when $b_0$ and $b_1$ have the same parity, we get  $(a_0^2 + a_1^2) + 2(b_0^2 + b_1^2) = 2 \md 4$ and $\sde(|x|^2)=2m-2$. In the worst case $2m-3\le \T_k(x) +2$, which gives us a bound on~$m$. 

To prove the second part of the proposition we consider a minimal unitary $U[x,y,k]$ and note that $|x|^2$ and $|y|^2$ can be expressed as 
$(x_0 + \sqrt{2} x_1)/2^m$ and $(y_0 + \sqrt{2} y_1)/2^m$, correspondingly.
The equality $|x|^2+|y|^2=1$ implies that $x_0+y_0 = 2^m$. Using the non-negativity of $y_0$ we get
\[
 x_0 = (a_0^2 + a_1^2) + 2(b_0^2 + b_1^2) \le 2^m,
\]
which leads to the desired bounds on $a_j$ and $b_j$.
\end{proof}

\begin{proof}[Proof of Proposition \ref{pr:delta}]
Firstly, we show that $|x-e^{i\theta}|\le\st\d$. To accomplish this, we expand $|x-e^{i\theta}|^2$ and note that the inequality $|x|^2\le1$ implies that $ |x-e^{i\theta}|^2 \le 2- 2\re(xe^{-i\theta})$.  We obtain the bound $2-2\re(xe^{-i\theta})\le2\d^2$ noting that $\sqrt{1-\left|\re(xe^{-i\theta})\right|} \le \d \le 1$.
 
Secondly, for any complex number $z$, the absolute values of its real and imaginary parts are both less than $|z|$.  Using the inequality $|x-e^{i\theta}|\le\st\d$, derived in the previous paragraph, we conclude that: 
\[
 \begin{array}{c}
  |(a_0+\sqrt{2}b_0)-\cos(\theta)\sqrt{2^m}|\le\sqrt{2^{m+1}}\d, \\
  |(a_1+\sqrt{2}b_1)-\sin(\theta)\sqrt{2^m}|\le\sqrt{2^{m+1}}\d.
 \end{array}
\]
\end{proof}

\begin{proof}[Proof of Proposition \ref{pr:split}]
First we show that $\re(xe^{-i\theta})>0$. Inequalities for real and imaginary parts of $\sqrt{2^m}(x-e^{i\theta})$ imply $|x-e^{i\theta}|\le2\d$ and $|x|\ge1-2\d$. Using that $ |x-e^{i\theta}|^2 = 1+|x|^2 - 2\re(xe^{-i\theta})$ we conclude that $\re(xe^{-i\theta})\ge1-2\d$ which is always non-negative when $\d\le1/2$. 

Second, we use $\re(xe^{-i\theta})>0$ to write
\[
\begin{array}{l}
 (\sqrt{1-\left|\re(xe^{-i\theta})\right|})^2 \sqrt{2^{m}}= \\
 \sqrt{2^{m}}(1-\cos(\theta)\re(x)-\sin(\theta)\im(x)).
\end{array}
\]
By replacing $1$ with $\cos(\theta)^2+\sin(\theta)^2$ we find 
\[
\begin{array}{l}
 (\sqrt{1-\left|\re(xe^{-i\theta})\right|})^2 \sqrt{2^{m}}= \\
 \cos(\theta)(\cos(\theta)\sqrt{2^{m}}-\sqrt{2^{m}}\re(x))+ \\
 \sin(\theta)(\sin(\theta)\sqrt{2^{m}}-\sqrt{2^{m}}\im(x)),
\end{array}
\]
which leads to the required result.
\end{proof}


\section{Norm equations}\label{sec:norm-eq}

\begin{figure}[t]
\begin{algorithmic}[1]
\Require $x'=a_0+b_0\sqrt{2}+i(a_1+b_1\sqrt{2}), m,k \in \Z, m \ge 0$
\Procedure{$\text{MIN-T-COUNT}$}{$x',m,k$}
  \State $a+b\st \gets 2^m - |x'|^2,s\gets\sde(|x'|^2/2^m)$
  \If{$s\le4$}
    \State \Return $\infty$
  \EndIf
  \If{$\text{IS-SOLVABLE}(a{+}\sqrt{2}b)$}
    \State \Return $s-2+(k+s)\mathrm{mod}\,2$
  \Else
    \State \Return $\infty$
  \EndIf
\EndProcedure
\Ensure $n$
\end{algorithmic}
\caption[$\text{MIN-T-COUNT}$ Procedure]{\label{alg:min-t-count}$\text{MIN-T-COUNT}$ Procedure. Outputs $\T_k(x'/\sqrt{2^m})$ if it is greater or equal to $4$ and $\infty$ otherwise.}
\end{figure}

In this section we discuss mathematical tools required to compute $\T_k(x)$~(procedure MIN-T-COUNT) and to enumerate all minimal exact unitaries with top-left entry $x$ -- unitaries of the form $U[x,y,k]$~(procedure ALL-UNITARIES). Alternatively, these two problems can be reformulated using equation 
\begin{equation}\label{eq:uni}
 |y|^2 = 1-|x|^2,
\end{equation}
that expresses the condition of $U[x,y,k]$ being a unitary matrix.  Proposition~\ref{prop:t-count} implies that it is easy to find $\T_k(x)$  when the equation above has a solution for some $y$ from $\Zf$. The problem of enumerating all unitaries is directly related to enumerating all the solutions of the equation (\ref{eq:uni}). 

We reduce the equation (\ref{eq:uni}) to the relative norm equation between two rings of algebraic integers: 
\[
 \Zw := \set{a_0+a_1\w+a_2\w^2+a_3\w^3}{a_i \in \Z }, \w := e^{i\pi/4},
\]
and its real subring $\Zs := \set{ a + b\sqrt{2}}{a,b\in\Z}$. Indeed, $x$ and $y$ can be expressed as $x'/\sqrt{2}^m$ and $y'/\sqrt{2}^m$ where $x'$ and $y'$ are from $\Zw$. Equation (\ref{eq:uni}) can then be rewritten as
\begin{equation}\label{eq:rn}
 |y'|^2 = 2^m - |x'|^2 = A+B\sqrt{2} \in \Zs,
\end{equation}
which is a special case of the relative norm equation well studied in the literature\cite{Coh2,Coh3}. The conditions required for this equation to be solvable, as well as the methods for enumerating its solutions (when exist) are well-known.  The algorithmic solution is furthermore available in the software package PARI/GP~\cite{PARI}. Next we give a brief simplified overview of the results related to solving the equation~(\ref{eq:rn}).

Symmetries in equations frequently provide useful insights about their solvability.  As was observed in \cite{S}, the automorphism $Aut$ of both $\Zw$ and its subring $\Zs$ defined as:
\begin{equation*}
Aut(\w):= -\w, Aut(\sqrt{2}):= -\sqrt{2}, Aut(a):= a,  \text{ for } a\in\Z
\end{equation*}
is useful in studying the equation (\ref{eq:rn}). It can be easily verified that the mapping $Aut(\cdot)$ preserves the addition, multiplication and commutes with complex conjugation and taking the norm squared ($|\cdot|^2$). Therefore, if $y'$ is a valid solution to the equation (\ref{eq:rn}), then $Aut(y')$ must also be a valid solution to the equation $Aut(A+B\sqrt{2})=A-B\sqrt{2}$. This implies the following necessary conditions for the equation (\ref{eq:rn}) to be solvable:
\begin{equation} \label{eq:nec-cond}
A+ B\sqrt{2} \ge 0, A- B\sqrt{2} \ge 0. 
\end{equation}
However, as we will show below, this condition is not sufficient. 

The problem of solving the relative norm equation can be reduced to a set of subproblems. Suppose that $A+B\sqrt{2}$ can be written as a product $(A_1+B_1\sqrt{2})(A_2+B_2\sqrt{2})$ and $\{y_j\}$ is the set of solutions to the relative norm equation (\ref{eq:rn}) with the right-hand side of $A_j+B_j\sqrt{2}$. In such case, $\w^k y_1 y_2$, $\w^k y_1^{\ast} y_2$, $\w^k y_1 y_2^{\ast}$, and $\w^k y_1^{\ast} y_2^{\ast}$ are also solutions of the equation (\ref{eq:rn}), for any integer $k$.  More generally, any element $A+B\sqrt{2}$ of $\Zs$ can be written as the product: 
\[
 u \sqrt{2}^{k(0)} p_1^{k(1)} \ldots p_N^{k(N)} q_1^{l(1)} \ldots q_M^{l(M)}, 
\]
also known as prime factorization, where $u$ is a unit in the ring $\Zs$, $p_{1..N}$ and $q_{1..M}$ are primes in the ring $\Zs$, and $k(\cdot), l(\cdot) \in \mathbb{N} \cup {0}$.  It turns out that every such $u$ possesses the property $u\cdot Aut(u)=1$ and $p_j$ are distinguished from $q_j$ as follows: $p_j\cdot Aut(p_j)$ are integer prime numbers of the form $8n\pm1$, whereas $q_j\cdot Aut(q_j)$ are integer prime numbers of the form $8n\pm3$.  The above factorization raises to the general theory of quadratic extensions and classification of primes into ramified, split and inert (see Chapter 3.4 of \cite{Coh1}).  When the necessary condition (\ref{eq:nec-cond}) holds, the right hand side can thus be rewritten as 
\begin{equation} \label{eq:decomposition}
  (\sqrt{2}-1)^{l(0)} \sqrt{2}^{k(0)} p_1^{k(1)} \ldots p_N^{k(N)} q_1^{l(1)} \ldots q_M^{l(M)}, 
\end{equation}
keeping in mind that $p_j^{k(j)}$, $Aut(p_j^{k(j)})$, $q_j^{l(j)}$, and $Aut(q_j^{l(j)})$ should all be positive.  $\sqrt{2}-1$ is the fundamental unit of $\Zs$ (for results on unit groups of quadratic extensions see Chapter 3.4.2 of \cite{Coh1}). The condition (\ref{eq:nec-cond}) implies that the sum $k(0)+l(0)$ must be an even number.  The relative norm equation (\ref{eq:rn}) with the right-hand side of $A+B\sqrt{2}$ is solvable if and only if each of the equations 
$|y'|^2 = p_j^{k(j)}$ and $|y'|^2 = q_j^{l(j)}$ is solvable.  Next we discuss the solvability of such equations with primes in $\Zr$ in more details. 

When $k(j)$ and $l(j)$ are even integers, the equations obviously have a solution.  Every odd number is a sum of an even number and a 1.  The equation $|y'|^2=p_j$ has a solution if and only if $p\cdot Aut(p)$ is a prime of the form $8n{+}1$. A probabilistic polynomial time algorithm for finding a solution can be found in \cite{S}.  In this case, according to the theory of cyclotomic number fields, it is said that the prime integer  $p\cdot Aut(p)$ splits completely~\cite{W}.  Solvability of the equation $|y'|^2=q_j$ is related to two subrings of $\Zw$: $\Z[i\sqrt{2}]$ and $\Z[i]$. When $q_j$ has the form $8n{-}3$, it splits in $\Z[i]$. In other words, there is a solution to the equation of the form $y'{=}a\pm bi$, where $a$ and $b$ are integers.  When $q_j$ has the form $8n{+}3$, it splits in $\Z[i\sqrt{2}]$---in other words, there is a solution to the equation of the form $y'{=}a\pm ib\sqrt{2}$, where $a$ and $b$ are integers.  In both cases, the equations can be solved in probabilistic polynomial time, see Chapter 4.8 \cite{Coh3} for details. 

When the ring factorization (\ref{eq:decomposition}) is known, the solution to the equation (\ref{eq:rn}) can be found efficiently.  The problem of factorization over the ring $\Zr$ can be reduced to that of factoring integers.  We note that the mapping/norm $\N: p \mapsto p\cdot Aut(p)$ is multiplicative and applying it to (\ref{eq:decomposition}) gives us integer
\[
 (-1)^{l(0)+k(0)} 2^{k(0)} \N(p_1)^{k(1)}\ldots \N(p_N)^{k(N)} q_1^{2 l(1)} \ldots q_M^{2 l(M)}.
\]
Therefore, via employing integer factoring we can find all $\N(p_j)$, $k(j)$, $q_j$, and $l(j)$.  Recall that $\N(p_j)$ are prime numbers of the form $8n\pm1$ and $q_j$ are prime numbers of the form $8n\pm3$. Numbers $p_j$ can be found in probabilistic polynomial time using the algorithm from Chapter~4.8~\cite{Coh3}, that is available as a part of PARI/GP package~\cite{PARI}. Therefore, the implementation of the predicate IS-SOLVABLE in the MIN-T-COUNT procedure~(Fig.~\ref{alg:min-t-count}) first checks the necessary conditions~(\ref{eq:nec-cond}).  If they hold, it computes the norm $\N$ of the right-hand side of the equation, and finds $\N(p_j)$, $k(j)$, $q_j$, and $l(j)$. The procedure returns TRUE if there is no $\N(p_j)$ of the form $8n-1$ such that $k(j)$ is odd, and FALSE otherwise. 

We demonstrate how to enumerate all solutions to~(\ref{eq:rn}) with the following example.  Consider the equation: 
\[
 |y'|^2 = 1\,828\,037\,034 - 1\,292\,617\,383 \sqrt{2}.
\]
The norm $\N$ of the right-hand side has prime integer factorization $2\cdot193\cdot2297\cdot3^2$.  We next find that the equation can be rewritten as  
\[
 |y'|^2 = (\sqrt{2}-1)^{15}\cdot\sqrt{2}\cdot(15-4\sqrt{2})\cdot(53-16\sqrt{2})\cdot3.
\]
The general form of the solution is
\[
 y' = (\sqrt{2}-1)^7 \w^k y_0{\cdot}Y_1{\cdot}Y_2{\cdot}Y_3,\, Y_j \in \{y_j,y_j^{\ast}\},
\]
where $y_0=1-\w$, $y_1=-1-3\w+\w^2-2\w^3$, $y_2=3-6\w-2\w^2+2\w^3$, and $y_3=1\pm i\sqrt{2}$.  We do not consider complex conjugation of $y_0$ because $(y_0)^{\ast}=\w^3 y_0$.  This is related to the fact that 2 is the only ramified prime in $\Zw$~\cite{W}. Taking into account that $k$ can take values from zero to seven we find that there are $64$ different solutions to the equation in our example.  

To implement ALL-UNITARIES procedure we need to factorize the right-hand side of the relative norm equation~(\ref{eq:rn}), find and record all possible solutions, write down all the unitaries, and pick those that are minimal. 

\section{Experimental results}\label{sec:exp}
In this section we discuss performance of our C++ implementation of the above algorithms.  We report memory and processing time required by our implementation, as well as the precision in approximation we were able achieve.  In our experiments we used a high performance server with eight Quad-Core AMD Opteron 8356 (2.30 GHz) processors and 128 GB of RAM memory.  Our implementation completely utilizes the processing power of the server and runs 32 threads in parallel.  The binary and the source code are available at \url{https://code.google.com/p/sqct/}.

To obtain the estimates for the time and memory required to run our algorithm to approximate some target unitary we found T-optimal approximations of $R_z$ rotations by angles of the form $2\pi k/1000$ for $k=1..1000$.  We used circuits with up to $109$ T gates for the approximation. The only known algorithm that gives the same optimality guarantee as our algorithm is the naive brute force search, \cite{F1}. Let $N_{\BFS}$ and $t_{\BFS}$ be the number of records and user time needed for the brute force search.  Since the number of unitaries with T-count at most $n$ scales as $192\cdot(3\cdot2^n-2)$~\cite{MA}, both $\log_2(N_{\BFS})$ and $\log_2(t_{\BFS})$ equal to $n$ up to an additive constant.  The most memory consuming part of our algorithm is the FIND-HALVES procedure.  The base two logarithm of the number of records it produces scales as $0.17n+5.07$ on average and as $0.25n+3.15$ in the worst case, see Fig.~\ref{fig:memory}. Fig.~\ref{fig:time} illustrates the fact that FIND-HALVES is the most time consuming part of the algorithm and, on average, logarithm base two of the time in milliseconds required to execute this step scales as $0.21n-10.41$.  Even though our algorithm requires an exponential amount of time and memory, the constants in the exponent are between four and five times better than those in the naive brute force search.  This allows us to find T-optimal approximations with precisions up to $10^{-15}$ using modern computers.  We believe such a precision to be sufficient for most applications, as we discussed in the introduction. 

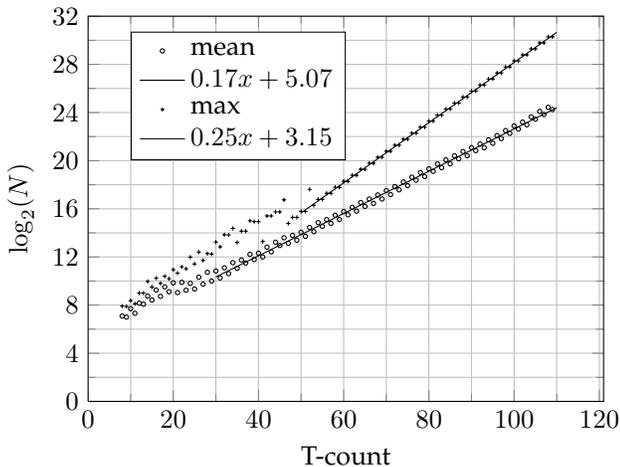
\begin{figure}
 \begin{tikzpicture}
\pgfplotsset{
every axis plot post/.append style={
every mark/.append style={scale=0.4}}}
legend style={
}
\begin{axis}[grid=both, xmin = 0, ymin =0, ymax = 32,ytick ={0,4,...,32},xtick ={0,20,...,120},minor tick num=1,
		xlabel=T-count,ylabel= $\log_2(N)$ , legend style={cells={anchor=west},anchor=north east,at={(0.5,0.96)} },legend columns=1,y post scale=0.9, x post scale=1 ]
\selectcolormodel{gray}

\addplot[color=black,only marks,mark=o] coordinates
{(8,7.0871657852203980242)
(9,7.0038270781114757483)
(10,7.7023805016155248507)
(11,7.3198346689772366811)
(12,8.1613485859633833657)
(13,8.0773979008509592358)
(14,8.7529158080514890527)
(15,8.4180634197968631779)
(16,9.2456857988351793664)
(17,8.7260818889022409749)
(18,9.5149921134661867175)
(19,9.1099168217260348381)
(20,9.8467173876648308959)
(21,9.0408896916272882853)
(22,9.8867452247421088342)
(23,9.2354877394584826161)
(24,9.8041600684420555047)
(25,9.3352161488001850797)
(26,10.3185495882069601706)
(27,9.7354781755098805804)
(28,10.7231890692839060889)
(29,9.9963873340340215648)
(30,10.8353398817123110711)
(31,10.256708767584378941)
(32,11.0968225291646715343)
(33,10.618670892704513115)
(34,11.525561427142136426)
(35,11.0509090848197641464)
(36,11.7444512540000132571)
(37,11.4963654747778558401)
(38,12.2150527292652303569)
(39,11.78048440432029623)
(40,12.3141249303025039231)
(41,11.9987843273386436938)
(42,12.8131823485066093007)
(43,12.4199738675898951762)
(44,13.2249097072476108108)
(45,12.9523037137833561801)
(46,13.5954906066076408613)
(47,13.13496979555556995)
(48,13.7888451440072604286)
(49,13.3946915140496932602)
(50,14.0496640326076092924)
(51,13.7128432538780982724)
(52,14.4494792251114919817)
(53,14.0989047087344077422)
(54,14.8546116335338817999)
(55,14.5363197736519683132)
(56,15.1042228799255092041)
(57,14.7843815951836510796)
(58,15.4678645382264225053)
(59,15.1115908023104664457)
(60,15.7743141124381573742)
(61,15.5036159507115119969)
(62,16.1211855700907223246)
(63,15.804009253349756812)
(64,16.5221167609917755555)
(65,16.1994590991181358588)
(66,16.8122330077732496274)
(67,16.4514178886696546369)
(68,17.1490469635422925137)
(69,16.8169062607581881323)
(70,17.5155519750980965834)
(71,17.1786928484633345731)
(72,17.844717770885499987)
(73,17.5766247225443401412)
(74,18.2288789987721775956)
(75,17.9293395463913137417)
(76,18.6256996371512414341)
(77,18.3726473873800214215)
(78,19.0217480401262574027)
(79,18.6804750145798167415)
(80,19.3245432166801121625)
(81,19.1125447020050670529)
(82,19.722833852844427638)
(83,19.4143878949672182488)
(84,20.0653930072645623546)
(85,19.7354015443449670489)
(86,20.410569764149450241)
(87,20.1526668599374272711)
(88,20.7591701256147870044)
(89,20.4090212862364564019)
(90,21.0828704131989885344)
(91,20.8333456004200183537)
(92,21.3816389351982492126)
(93,21.0802885255802106382)
(94,21.7290857550657428656)
(95,21.4462280191575217599)
(96,22.0857678474145034042)
(97,21.8392907508630215084)
(98,22.5001588815296495571)
(99,22.2749027224631659012)
(100,22.887902901351082638)
(101,22.6258088383339821404)
(102,23.2031622065194945678)
(103,22.9800013851662936352)
(104,23.6412351617613346312)
(105,23.4479410792043839527)
(106,24.0909545453586762894)
(107,23.8355325921954763091)
(108,24.4304865450706718844)
(109,24.2511110229792181782)
};

\addplot[domain=30:110]{0.175678531948060050042*x+5.06837111023308728516};

\addplot[color=black,only marks,mark=+] coordinates
{(8,7.9188632372745945124)
(9,7.8826430493618412588)
(10,8.3706874068072176976)
(11,8.1085244567781690537)
(12,9.)
(13,9.)
(14,9.971543553950771991)
(15,9.4998458870832053605)
(16,10.2372099607550214697)
(17,9.8105716347411469424)
(18,10.3987436919381932401)
(19,10.1910592145316563319)
(20,10.945443836377911529)
(21,10.6564248632777803261)
(22,11.1823943534045288638)
(23,11.0147179298600095167)
(24,11.982280604558282924)
(25,11.4152135987652863246)
(26,12.4062050054188541311)
(27,11.7168194613299490274)
(28,12.2839564876651142536)
(29,12.2407913321619568603)
(30,13.2363132270428051154)
(31,12.8394006489884125195)
(32,13.841269488914093002)
(33,13.8335819305416933272)
(34,14.3758539764184670259)
(35,13.2088439907346147138)
(36,14.1446582428318823208)
(37,14.1446582428318823208)
(38,15.0071587017813995441)
(39,14.9385688271438764102)
(40,14.933690654952233734)
(41,13.2875681028314043563)
(42,15.4075997158904035152)
(43,15.4075997158904035152)
(44,15.7438614141974452277)
(45,15.7438614141974452277)
(46,16.7436510743030099841)
(47,14.7873924667365086864)
(48,15.287351660697491677)
(49,15.287351660697491677)
(50,15.7873669573688799579)
(51,15.7873669573688799579)
(52,17.6108734516740822093)
(53,16.287351660697491677)
(54,16.7873797121090755928)
(55,16.7873797121090755928)
(56,17.287378717740110018)
(57,17.287378717740110018)
(58,17.7873669573688799579)
(59,17.7873669573688799579)
(60,18.2873742082682458179)
(61,18.2873742082682458179)
(62,18.7873701460645004639)
(63,18.7873701460645004639)
(64,19.2873742082682458179)
(65,19.2873742082682458179)
(66,19.7873733347530732176)
(67,19.7873733347530732176)
(68,20.2873753356375333106)
(69,20.2873753356375333106)
(70,20.7873741319241151984)
(71,20.7873741319241151984)
(72,21.2873730808980773612)
(73,21.2873730808980773612)
(74,21.7873733347530732176)
(75,21.7873733347530732176)
(76,22.2873733627407020658)
(77,22.2873733627407020658)
(78,22.7873733347530732176)
(79,22.7873733347530732176)
(80,23.2873733627407020658)
(81,23.2873733627407020658)
(82,23.7873734343994775541)
(83,23.7873734343994775541)
(84,24.2873732922800510516)
(85,24.2873732922800510516)
(86,24.7873733845762762462)
(87,24.7873733845762762462)
(88,25.2873733275103769888)
(89,25.2873733275103769888)
(90,25.7873733347530732176)
(91,25.7873733347530732176)
(92,26.2873733098952141277)
(93,26.2873733098952141277)
(94,26.7873732849298684685)
(95,26.7873732849298684685)
(96,27.2873732746648877603)
(97,27.2873732746648877603)
(98,27.7873732662461662439)
(99,27.7873732662461662439)
(100,28.287373265857306034)
(101,28.287373265857306034)
(102,28.7873732600182654486)
(103,28.7873732600182654486)
(104,29.287373263655410594)
(105,29.287373263655410594)
(106,29.78737326157524065)
(107,29.78737326157524065)
(108,30.287373263655410594)
(109,30.287373263655410594)
};

\addplot[domain=50:110]{0.249999982592102979115*x+3.15946798932939542048};

\legend{mean$\quad$,$0.17x+5.07$,max$\quad$,$0.25x+3.15$, MIN-T-COUNT}

\end{axis}%
\end{tikzpicture}%
 \caption{\label{fig:memory}Scaling of the number of records found by FIND-HALVES procedure when executing $RCUP$ algorithm for a given T-count. Graph shows the average and maximum number of records per a given T-count. The data is averaged over T-optimal approximations of the $R_z$ rotations by angles $2\pi k/1000$ for $k=1..1000$.}
\end{figure}

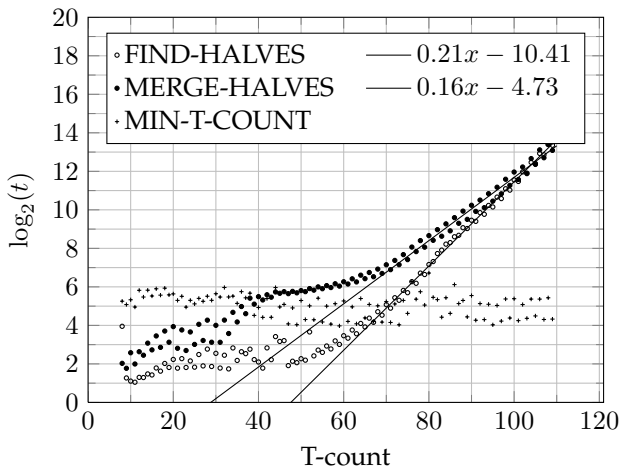
\begin{figure}
 \begin{tikzpicture}
\pgfplotsset{
every axis plot post/.append style={
every mark/.append style={scale=0.4}}}
legend style={
}
\begin{axis}[grid=both, xmin = 0, ymin =0, ymax = 20,ytick ={0,2,...,20},xtick ={0,20,...,120},minor tick num=1,
		xlabel=T-count,ylabel= $\log_2(t)$ , legend style={cells={anchor=west},anchor=north,at={(0.5,0.96)} },legend columns=2,y post scale=0.9, x post scale=1 ]
\selectcolormodel{gray}

\addplot[color=black,only marks,mark=o] coordinates
{(8,3.9524123159845526295)
(9,1.2698380773947703015)
(10,1.1070202592052907067)
(11,1.0482975756250229132)
(12,1.3007434376851527944)
(13,1.2912393022727890278)
(14,1.4788284341330231854)
(15,1.4259254466122742712)
(16,1.7880642366249124483)
(17,1.6242164928110550471)
(18,2.0227491667997681301)
(19,1.8148354480561598645)
(20,2.2315267212462394381)
(21,1.7789818723742000903)
(22,2.2735048493875851656)
(23,1.8140963183386707116)
(24,2.1539275547290410012)
(25,1.8305664680869624015)
(26,2.4874549637013321503)
(27,1.8954430057944965838)
(28,2.7689299921423423364)
(29,1.8654006178475291287)
(30,2.5480883629387040397)
(31,1.7846999583785997901)
(32,2.4507134858418552035)
(33,1.7403348221783924496)
(34,2.8285058684082982895)
(35,2.1798438190248745235)
(36,2.6323489576237968202)
(37,2.5321037189831619161)
(38,2.7643068724010985432)
(39,2.2233150388896730694)
(40,2.0992679325984937107)
(41,1.7809890769096262962)
(42,2.8317688323649034037)
(43,2.219725506230866418)
(44,3.4176196400502971437)
(45,3.1666968086963053995)
(46,3.2077748081270504903)
(47,1.9113997541061403398)
(48,2.1386093845951734966)
(49,2.0642931018739851153)
(50,2.2693680403110235707)
(51,2.1941296135143879869)
(52,3.6593797925298864347)
(53,2.4158884383506191811)
(54,2.7190125416535812735)
(55,2.5990177962136576563)
(56,2.9287692002947774534)
(57,2.776234983047852295)
(58,3.1943447658091762772)
(59,3.0119886185309545904)
(60,3.4632079755827132282)
(61,3.3358886639303771917)
(62,3.759241817567176937)
(63,3.5666569967464971928)
(64,4.1011330065177471588)
(65,3.9195660520044688815)
(66,4.3761760357760639296)
(67,4.1593769020137235659)
(68,4.7112137739860356844)
(69,4.5296695492873194088)
(70,5.0760053230037170579)
(71,4.8944687101204536704)
(72,5.436756342014898773)
(73,5.3035732408361909773)
(74,5.81573303571125772)
(75,5.6983231198362438426)
(76,6.2705842757757375621)
(77,6.1860191696170906294)
(78,6.7146862967094482611)
(79,6.640988712285539001)
(80,7.1606203146941339238)
(81,7.3139290288976920503)
(82,7.7303849281987406851)
(83,7.907953386137212115)
(84,8.2186740415693682056)
(85,8.2997994204410877293)
(86,8.6028127805453236149)
(87,8.6780048654331911457)
(88,9.0721053689369885067)
(89,9.0295273956824720363)
(90,9.4589316012791361781)
(91,9.4025917928973608723)
(92,9.8544726223489853418)
(93,9.768682676594784907)
(94,10.228719992352125157)
(95,10.1496389450706644623)
(96,10.660893287905587966)
(97,10.5844041326313732194)
(98,11.0954448999562175938)
(99,11.0234791347030673171)
(100,11.5555041402465775688)
(101,11.4754789688924044395)
(102,11.9713301978498781128)
(103,11.9048534743981669565)
(104,12.4658118512463747612)
(105,12.4080047824975342703)
(106,12.9288919552918586276)
(107,12.8553662656950070317)
(108,13.4002342196672052154)
(109,13.3512702872134732819)
};

\addplot[domain=0:110]{0.218698536279075497613*x+-10.4100066768000459237};

\addplot[color=black,only marks,mark=*] coordinates
{(8,2.0293127532721396536)
(9,1.7714578407359261176)
(10,2.5791023276468753911)
(11,1.9953448907179429911)
(12,2.6328078002692567553)
(13,2.4511122525408698745)
(14,3.0783845555107923881)
(15,2.7277206101469511467)
(16,3.4519437597146051366)
(17,2.8591188287202056759)
(18,3.7065425833821482474)
(19,3.1257682446354569149)
(20,3.9359558618944716734)
(21,2.8113831533396721419)
(22,3.7623200118491029585)
(23,2.6961617987105100676)
(24,3.6584721552801239168)
(25,3.0125611376811364837)
(26,4.0562921407827840629)
(27,3.2161185628388011291)
(28,4.271261114758262089)
(29,3.1255471511438781994)
(30,4.0028286534822413828)
(31,3.1224035621881117758)
(32,4.2747931394344984439)
(33,3.5779697624505600441)
(34,4.6766834229284091307)
(35,4.1769026114494232018)
(36,4.9788842277529069809)
(37,4.6073706512941576681)
(38,5.4124054890114516866)
(39,5.1035060317747147857)
(40,5.4915762728277055396)
(41,5.3199894883380169748)
(42,5.5911118365673860165)
(43,5.4636121959869374937)
(44,5.7363666112758410783)
(45,5.6661465425495161499)
(46,5.7520189485040030159)
(47,5.6478910511963266219)
(48,5.7657877507694624554)
(49,5.6897730504877198804)
(50,5.807368853091618888)
(51,5.7517231500933857332)
(52,5.9104364059071363044)
(53,5.8295054889352579853)
(54,5.9994928550073101287)
(55,5.9036754625271541963)
(56,6.0652001775037042974)
(57,5.9647376732471545926)
(58,6.1624546397752200284)
(59,6.0396467236107842957)
(60,6.2647623596366433127)
(61,6.1484366040795672582)
(62,6.4184348375318829503)
(63,6.2561367163930963733)
(64,6.6073037696183075351)
(65,6.4258800521472081064)
(66,6.7461872953881094008)
(67,6.5259092307831202919)
(68,6.922830972947151371)
(69,6.6974384799153583355)
(70,7.1474904481098605705)
(71,6.8934363765983929238)
(72,7.3732260549532162655)
(73,7.1517412956601500556)
(74,7.6749047638078564006)
(75,7.4151687086746290673)
(76,8.0677173120814509187)
(77,7.8315595746377085897)
(78,8.4029003946977388094)
(79,8.0622631028310370953)
(80,8.6611619728094845728)
(81,8.4176861581780763798)
(82,8.9709344556856889345)
(83,8.6309046628297016057)
(84,9.2715221226624636604)
(85,8.9088739849811124118)
(86,9.5945881098045588805)
(87,9.3014250343912912118)
(88,9.9339218842448252217)
(89,9.506054415187768488)
(90,10.237372107828100221)
(91,9.9197284852877174118)
(92,10.5068103758828502393)
(93,10.113600248409274452)
(94,10.8386332722759376179)
(95,10.4528208583577030769)
(96,11.1750203320740132588)
(97,10.8322701918174569876)
(98,11.5888559868333790244)
(99,11.2690204102658720964)
(100,11.9626693269775503556)
(101,11.5780237570088959182)
(102,12.2249136711057184566)
(103,11.8815357962190818187)
(104,12.6621113001697148246)
(105,12.3582647996250826717)
(106,13.1138846009731080768)
(107,12.7047836059373708297)
(108,13.3897962493466019106)
(109,13.088984865514374878)
};

\addplot[domain=0:110]{0.164112869068143737881*x+-4.72609735019904375714};

\addplot[color=black,only marks,mark=+] coordinates
{(8,5.2555165823504259899)
(9,5.0811433277160896911)
(10,5.330111883258135948)
(11,4.9528043832876559675)
(12,5.8212021752175752635)
(13,5.4729178944996093796)
(14,5.8279345182477691878)
(15,5.5265664412837066959)
(16,5.900627161719592799)
(17,5.5703291579092875946)
(18,5.9330998925641385413)
(19,5.5851008057160027974)
(20,5.6353220354357476256)
(21,5.2983686300457377817)
(22,5.5086278297087253744)
(23,4.9484645698674345218)
(24,5.3002644424645382684)
(25,5.1166120832428424889)
(26,5.4022265281836056629)
(27,5.2503973358452085435)
(28,5.5138317967219229501)
(29,5.3019759528123762104)
(30,5.6441619357516811996)
(31,5.471758413443302264)
(32,5.9671173412513947045)
(33,5.475271678650902672)
(34,5.5184276813919595607)
(35,5.1422109891703808402)
(36,5.7515644103850749749)
(37,5.431673000608187706)
(38,5.6721835737196943832)
(39,4.5181959704840667136)
(40,4.9284816043847941937)
(41,4.415625408514847692)
(42,5.3919721374978752655)
(43,4.4193610600982721471)
(44,5.9638832214936480392)
(45,4.9577489894149452093)
(46,5.088996638341587545)
(47,4.0778774206281332134)
(48,4.9071739956477128917)
(49,4.0364687388861967748)
(50,5.2922519454602781713)
(51,4.304272320353849999)
(52,5.0141612650780838254)
(53,4.2919008302810941435)
(54,5.4102984215405241354)
(55,4.1768452497734739183)
(56,4.9987025918869310225)
(57,4.0656387844075055087)
(58,5.1276719916334656587)
(59,3.9738688761623017758)
(60,5.2676829183826555924)
(61,4.206535051554640808)
(62,4.9628204273240231998)
(63,4.0827243160143271265)
(64,5.3948810260452662004)
(65,4.3875401324341002837)
(66,5.0874331334773768752)
(67,4.2025304997147122129)
(68,5.0176545573573209051)
(69,4.1694834041004229401)
(70,5.2345557290046987871)
(71,4.1514913683614900336)
(72,5.1128977732683474005)
(73,4.0315548258145718913)
(74,5.2955943212195358523)
(75,4.6223060657696754076)
(76,6.2817612127763622603)
(77,5.3601011545633325164)
(78,5.748112438908250789)
(79,5.0551509357739340771)
(80,6.718699349108692231)
(81,5.4550619658667945192)
(82,5.3384990876876853557)
(83,4.4253476005428296545)
(84,5.4363033403789133373)
(85,4.6228397820179321065)
(86,6.118467364780361394)
(87,5.1343842719788262578)
(88,5.2942167805125936599)
(89,4.4414235559651214152)
(90,5.5486147807992450405)
(91,4.2667561746799996415)
(92,5.1325659692816197876)
(93,4.3009490450595454928)
(94,5.0964038632597867384)
(95,4.0294234676718028727)
(96,5.0347387243650021973)
(97,4.2226787403743145706)
(98,5.3713021567973920839)
(99,4.3820668133464899536)
(100,5.385200757474133136)
(101,4.24611668632486991)
(102,5.0820143480386263906)
(103,4.3658689678165906399)
(104,5.3661581976718209765)
(105,4.4908481256916761799)
(106,5.364667153531224723)
(107,4.3056564888284823466)
(108,5.4324550890745956233)
(109,4.3273308585665763669)
};

\legend{FIND-HALVES$\quad$,$0.21x-10.41$,MERGE-HALVES$\quad$,$0.16x-4.73$, MIN-T-COUNT}
\end{axis}%
\end{tikzpicture}%
 \caption{\label{fig:time}Average user time $t$ (in milliseconds) required to run different parts of $RCUP$ algorithm for a given T-count.  Each point on the graph was obtained by averaging over T-optimal approximations of $R_z$ rotations by angles $2\pi k/1000$ for $k=1..1000$. }
\end{figure}

On average, the number of $T$ gates needed to achieve a given quality of approximation $\ve$ scales as $3.067\log(1/\ve)-4.322$~(Fig.~\ref{fig:cost}).  The knowledge of this scaling is important for estimating the resources required to run quantum algorithms having $R_z$ rotations as their building block.  The best previous estimate was based on the values of $\log_2(1/\ve)$ less than $14$, due to the inefficiency of the naive brute force approach~\cite{F1}.  Our result, summarized in Fig.~\ref{fig:cost}, results in a much more substantiated numerical evidence for the resource scaling. 

\begin{figure}
 \begin{tikzpicture}
\pgfplotsset{
every axis plot post/.append style={
every mark/.append style={scale=0.4}}}
legend style={
}
\begin{axis}[grid=both, xmin = 0, ymin =0, ymax = 150,ytick ={0,20,...,160},xtick ={0,5,...,40},minor tick num=1,
		xlabel=$x{=}\log_{2}(1/\ve)$,ylabel=T-count, legend style={cells={anchor=west},anchor=north,at={(0.5,0.96)} },legend columns=1,y post scale=0.9, x post scale=1 ]
\selectcolormodel{gray}

\addplot[color=black,only marks,mark=*] coordinates
{(4.563952654270077,8)
(4.804911726543429,9)
(5.114512247350719,10)
(5.321699917841492,11)
(5.389381320279023,12)
(5.698633575565129,13)
(5.9851189732734715,14)
(6.219584162577563,15)
(6.628885747307054,16)
(6.944115070272085,17)
(7.2381847679633715,18)
(7.549987534103624,19)
(8.060025759508576,20)
(8.360370164705174,21)
(8.778857270721927,22)
(9.268359849189087,23)
(9.513019579323734,24)
(9.719470550266257,25)
(9.985578903336368,26)
(10.244138012301264,27)
(10.618780271260764,28)
(10.981849008035145,29)
(11.260952646302643,30)
(11.543235584294147,31)
(11.822171576670671,32)
(12.130630677251684,33)
(12.425402679332667,34)
(12.7887190121359,35)
(13.045899150470058,36)
(13.301111053921307,37)
(13.622650272714983,38)
(14.026577667616742,39)
(14.304565193883864,40)
(14.60342270640839,41)
(14.938661864519771,42)
(15.234167631511182,43)
(15.49610636932265,44)
(15.816206512926959,45)
(16.174536907202086,46)
(16.517904500710575,47)
(16.919869464172933,48)
(17.26159248226288,49)
(17.606166198818826,50)
(17.925315655222587,51)
(18.223151402022946,52)
(18.46089652929471,53)
(18.787534534583386,54)
(19.217688580888442,55)
(19.547565213132206,56)
(19.85804629553689,57)
(20.226822399957324,58)
(20.557989331999853,59)
(20.839259606083075,60)
(21.216848673733672,61)
(21.548107426584572,62)
(21.819974525136505,63)
(22.158053575870113,64)
(22.539469613732265,65)
(22.920674700965662,66)
(23.211216381794788,67)
(23.56428739529473,68)
(23.852514383130806,69)
(24.21284462095379,70)
(24.534518951555018,71)
(24.823268717423005,72)
(25.158779040241278,73)
(25.483799783747425,74)
(25.770094723551875,75)
(26.045996658345622,76)
(26.38290822443594,77)
(26.75871332008588,78)
(27.09896310620405,79)
(27.334663992680646,80)
(27.711561316923067,81)
(28.058798732191978,82)
(28.387635965531157,83)
(28.76502109912516,84)
(29.063158579274305,85)
(29.36173241257622,86)
(29.737407527903862,87)
(30.15250003877172,88)
(30.44407376128963,89)
(30.74487275263088,90)
(31.18580088132716,91)
(31.562356369088196,92)
(31.873973924176493,93)
(32.236857672211215,94)
(32.555325458751426,95)
(32.879800921725675,96)
(33.1660671043651,97)
(33.46767079112834,98)
(33.81437791102503,99)
(34.17705693461414,100)
(34.56733203592372,101)
(34.89114579457451,102)
(35.154976315307586,103)
(35.43964513206496,104)
(35.72732770038344,105)
(36.1143369052569,106)
(36.46628563572319,107)
(36.750857547948975,108)
(37.11368317301026,109)
};

\addplot[domain=0:40] {3.06672*x+-4.32208};

\legend{mean,$3.067x-4.322$}
\end{axis}%
\end{tikzpicture}%
 \caption{\label{fig:cost}Scaling of the T-count required to achieve the given approximation precision $\ve$. Each point in the graph shows the average precision that can be achieved when using circuits with the given T-count. The average is taken over T-optimal approximations of $R_z$ rotations by angles $2\pi k/1000$ for $k=1..1000$. }
\end{figure}
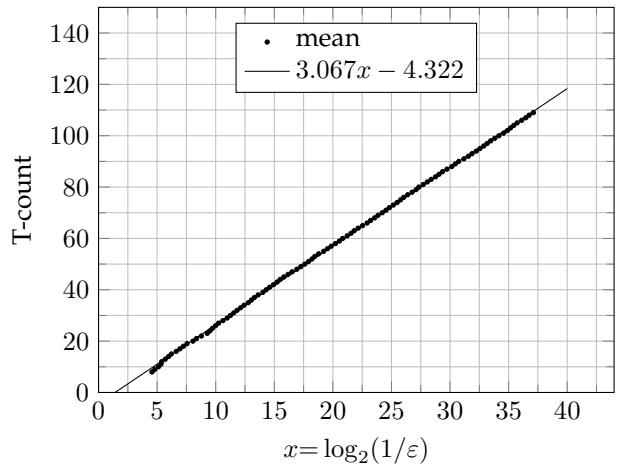

We computed optimal circuits for $R_z$ rotations by angles $\pi/2^k$ for $k=3,4,\ldots,27$. These rotations are used in the Quantum Fourier Transform, latter being a common building block for many quantum algorithms.  We found approximations reaching precision up to $10^{-15}$~(Fig.~\ref{fig:qft-cost}).  To make a direct comparison to~\cite{S} we also found optimal approximations of $R_z(0.1)$ using up to 153 T gates and reaching precision $3.18\cdot10^{-16}$~(Fig.~\ref{fig:01-cost}). Computing these approximations of $R_z(0.1)$ took 33.2 hours in total, user time.  Our circuit approximations are about $25\%$ shorter than those obtained using the algorithm from~\cite{S}, and are guaranteed to be optimal. 

All optimal circuits found by our algorithm are available online at \url{https://code.google.com/p/sqct/}. 

\begin{figure}
\begin{tikzpicture}
\pgfplotsset{
every axis plot post/.append style={
every mark/.append style={scale=0.4}}}
legend style={
}
\begin{axis}[grid=both, xmin = 0, ymin =0, ymax = 160,ytick ={0,20,...,160},xtick ={0,5,...,50},minor tick num=1,
		xlabel=$x{=}\log_{2}(1/\ve)$,ylabel=T-count, legend style={cells={anchor=west},anchor=north west,at={(0.04,0.96)} },legend columns=1,y post scale=0.9, x post scale=1 ]
\selectcolormodel{gray}

\addplot[color=black,only marks,mark=+] coordinates
{(4.8220783787517612783,0)
(4.8220783787517403233,1)
(4.8220783787517403233,2)
(4.8220783787517403233,3)
(4.8220783787517403233,4)
(4.8220783787517403233,5)
(4.8220783787517403233,6)
(4.8220783787517403233,7)
(4.8220783787517403233,8)
(4.8220783787517403233,9)
(4.8220783787517403233,10)
(5.4507107576468285283,11)
(5.4507107576468285283,12)
(5.4507107576468285283,13)
(5.4507107576468285283,14)
(5.7939969605361974849,15)
(6.8938308572119152835,16)
(6.8938308572119152835,17)
(7.2882698202291964021,18)
(8.0813011720599604508,19)
(8.0813011720599604508,20)
(8.0813011720599604508,21)
(10.7828518316750171181,22)
(10.7828518316750171181,23)
(10.7828518316750171181,24)
(10.7828518316750171181,25)
(10.7828518316750171181,26)
(10.7828518316750171181,27)
(10.7828518316750171181,28)
(10.7828518316750171181,29)
(10.7828518316750171181,30)
(10.7828518316750171181,31)
(10.7828518316750171181,32)
(11.4091189869906281437,33)
(14.2132179676906141125,34)
(14.2132179676906141125,35)
(14.2132179676906141125,36)
(14.2132179676906141125,37)
(14.2132179676906141125,38)
(14.2132179676906141125,39)
(14.2132179676906141125,40)
(14.3680956484066627368,41)
(14.3680956484066627368,42)
(15.1256253341502328286,43)
(15.1742039074049770581,44)
(15.692431044699169021,45)
(16.2054617433031798624,46)
(16.2054617433031798624,47)
(16.322589065407896129,48)
(17.8710769564268143265,49)
(17.8710769564268143265,50)
(17.8710769564268143265,51)
(17.8710769564268143265,52)
(17.9046578139035364311,53)
(18.4993853007177694529,54)
(18.4993853007177694529,55)
(19.2233496753083336367,56)
(19.4110915606160810316,57)
(19.9127533800558694039,58)
(19.9190637306140932426,59)
(20.6261190474897985839,60)
(20.8446240991123016404,61)
(21.0339064485696076023,62)
(24.6771268054120897156,63)
(24.6771268054120897156,64)
(24.6771268054120897156,65)
(24.6771268054120897156,66)
(24.6771268054120897156,67)
(24.6771268054120897156,68)
(24.6771268054120897156,69)
(24.6771268054120897156,70)
(24.6771268054120897156,71)
(24.6771268054120897156,72)
(24.9966292192926952901,73)
(24.9966292192926952901,74)
(25.3918807425696668184,75)
(25.6181913019071107113,76)
(25.6181913019071107113,77)
(27.4636629508497518722,78)
(27.4636629508497518722,79)
(27.5139999913356947762,80)
(29.344976073025341365,81)
(30.1091969883158077497,82)
(30.1091969883158077497,83)
(30.1091969883158077497,84)
(30.1091969883158077497,85)
(30.1091969883158077497,86)
(30.1091969883158077497,87)
(30.3394643301065359042,88)
(31.2083314634160049616,89)
(31.2083314634160049616,90)
(31.2083314634160049616,91)
(31.257311328619367317,92)
(31.257311328619367317,93)
(31.9191111537896898549,94)
(31.9191111537896898549,95)
(33.2406072488529525857,96)
(33.6011451780976015759,97)
(33.6011451780976015759,98)
(33.6011451780976015759,99)
(34.3954840456054201262,100)
(34.3954840456054201262,101)
(34.7866289531466106898,102)
(34.7866289531466106898,103)
(35.0294167256819910876,104)
(35.2416909057247291807,105)
(35.8583421080258836713,106)
(37.243215075389929536,107)
(37.243215075389929536,108)
(37.243215075389929536,109)
(37.243215075389929536,110)
(37.565900596438576428,111)
(38.067268294328434103,112)
(38.362550248271985236,113)
(40.09328630889585972,114)
(40.09328630889585972,115)
(40.09328630889585972,116)
(40.09328630889585972,117)
(40.09328630889585972,118)
(40.54054267617840099,119)
(40.54054267617840099,120)
(40.54054267617840099,121)
(40.54054267617840099,122)
(41.15328855545313798,123)
(43.1076296922912267,124)
(43.1076296922912267,125)
(43.1076296922912267,126)
(43.1076296922912267,127)
(43.1076296922912267,128)
(43.203454535570536,129)
(45.37467614309108,130)
(45.37467614309108,131)
(45.37467614309108,132)
(45.37467614309108,133)
(45.37467614309108,134)
(45.37467614309108,135)
(47.08114605594605,136)
(47.08114605594605,137)
(47.08114605594605,138)
(47.23611800623944,139)
(47.23611800623944,140)
(47.38823969748175,141)
(47.72860166791135,142)
(47.855503172726614,143)
(48.67029523164469,144)
(49.43232547456254,145)
(49.43232547456254,146)
(49.43232547456254,147)
(51.124087036840905,148)
(51.124087036840905,149)
(51.124087036840905,150)
(51.3893473315558,151)
(51.3893473315558,152)
(51.4816265872583,153)
};

\addplot[domain=5:153]{3.07342*x+-5.76405};

\legend{$R_z(0.1)$,3.07x-5.76}
\end{axis}%
\end{tikzpicture}%
\caption{\label{fig:01-cost}Scaling of the T-count required to achieve the given precision $\ve$ in approximating $R_z(0.1)$. }
\end{figure}
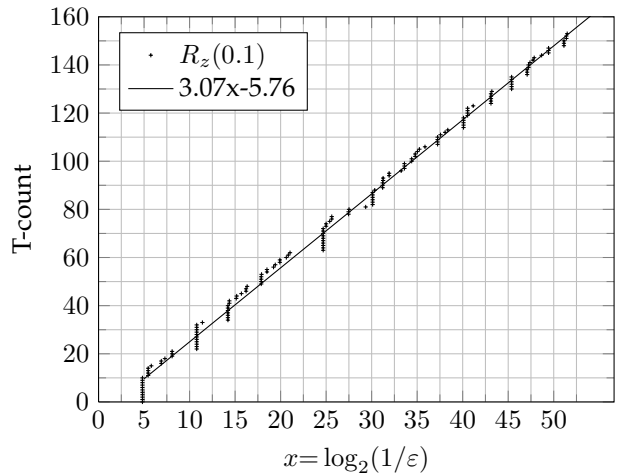

\begin{figure}
\begin{tikzpicture}
\pgfplotsset{
every axis plot post/.append style={
every mark/.append style={scale=0.4}}}
legend style={
}
\begin{axis}[grid=both, xmin = 0, ymin =0, ymax = 160,ytick ={0,20,...,160},xtick ={0,5,...,50},minor tick num=1,
		xlabel=$x{=}\log_{2}(1/\ve)$,ylabel=T-count, legend style={cells={anchor=west},anchor=north east,at={(0.25,0.96)} },legend columns=1,y post scale=0.9, x post scale=1 ]
\selectcolormodel{gray}

\addplot[color=black,only marks,mark=+] coordinates
{(6.4936844305623524623,0)
(6.493684439044377806,1)
(6.493684439044377806,2)
(6.493684439044377806,3)
(6.6400433104186028815,4)
(6.6400433104186077539,5)
(6.6400433104186077539,6)
(7.1296207538085950195,7)
(7.1296207538085950195,8)
(7.4535473007409629648,9)
(7.6268354848408353627,10)
(7.8878812495020941609,11)
(7.8878812495020941609,12)
(8.2086816187702973873,13)
(8.2086816187702973873,14)
(8.465887515961057935,15)
(8.9867240375305734634,16)
(9.1631975807952200553,17)
(9.4389599843233865793,18)
(9.6804114607693108441,19)
(9.7703445264546151961,20)
(10.1085069147585238439,21)
(10.8739218225129247381,22)
(11.1321428238567494311,23)
(11.250783082190485419,24)
(11.3489153032849493774,25)
(11.4472062267974602538,26)
(11.4950228510430441688,27)
(11.5711414171048755504,28)
(11.699983458216799817,29)
(12.0361554741780340089,30)
(12.259822441267770347,31)
(12.4694586961299447288,32)
(12.8790382194249043812,33)
(13.0079287594337813127,34)
(13.3322000587585177867,35)
(13.7643493555644049511,36)
(13.9117640929517345105,37)
(14.442024278976996121,38)
(14.6313570686607515499,39)
(14.7240407500738438181,40)
(14.9872783323295903286,41)
(15.2833636465116623582,42)
(15.4499746045631105536,43)
(15.7844738658171948669,44)
(16.0058157006859040091,45)
(16.2085532957722437611,46)
(16.6723268041635496914,47)
(16.8149328031755687175,48)
(17.0586440220198181454,49)
(17.6195313812741705333,50)
(17.7395674310459365427,51)
(17.8702933763581122108,52)
(18.4277569079777534587,53)
(18.4777005472145593149,54)
(18.6416561575516670041,55)
(19.0107931466528963395,56)
(19.1827369348237803087,57)
(19.3732508264179724236,58)
(19.4392839179379944135,59)
(19.8763924700064930432,60)
(20.3241389108510261533,61)
(20.5086151753119142565,62)
(21.0533023525955218092,63)
(21.456126165285793326,64)
(21.6943025859934981857,65)
(21.7968741520692065614,66)
(21.8514763736845617261,67)
(22.1190201161289274019,68)
(22.332184435483832599,69)
(22.8596642505362889332,70)
(22.9236043599293281214,71)
(22.9973682379732669592,72)
(23.6331512897420202625,73)
(23.7601374658622872024,74)
(23.9029744524956755733,75)
(24.1935044291295016494,76)
(24.3835115149661237872,77)
(24.6322069484427067976,78)
(24.7596766251326539468,79)
(25.0943812723287961905,80)
(25.754315062281060206,81)
(25.8088697346822554675,82)
(25.8801584527259984759,83)
(26.5950546814923511315,84)
(26.7391061053515086475,85)
(26.7989625775675105211,86)
(27.1582195493445647295,87)
(27.4823347147292486998,88)
(28.1070032652446104485,89)
(28.2894107964032610588,90)
(28.5256013891726696935,91)
(28.6036263975959761319,92)
(29.1041298035324465242,93)
(29.199091543909009695,94)
(29.780815696563541237,95)
(29.9177236525398094662,96)
(29.9454416233660367005,97)
(30.3673681944427584193,98)
(31.0306224500367809041,99)
(31.545791952063120956,100)
(31.5848081717969614293,101)
(32.3702043565249961203,102)
(32.3878792769520441942,103)
(32.9331353731674546134,104)
(33.18559797028282866,105)
(33.5581222905265288403,106)
(33.8001767382862204078,107)
(35.4648812754027721852,108)
(35.6464694107406253119,109)
(36.384327415405903886,110)
(37.241256571307389765,111)
(37.359859824165327369,112)
(37.520820292715374496,113)
(38.198140168843326404,114)
(38.483815803138089277,115)
(38.640562976937315364,116)
(39.120051621565866271,117)
(39.4561804517270388,118)
(39.59704203508333664,119)
(40.44269543238921682,120)
(40.6874044957674871,121)
(40.98262798459013983,122)
(41.42803562938946403,123)
(41.59270876142357,124)
(42.06557745700145,125)
(42.69348431934689,126)
(43.008181262484065,127)
(43.40225975424428,128)
(43.51672022711217,129)
(44.10185006758178,130)
(44.36431390640566,131)
(44.79978736447028,132)
(45.04769309867835,133)
(45.30926076906457,134)
(45.59252340157704,135)
(45.8919953714204,136)
(46.2805659460181,137)
(46.72908084035158,138)
(47.04941837444022,139)
(47.3641724398775,140)
(47.618185099035315,141)
(47.89369669626944,142)
(48.30936927745665,143)
(48.67329203578238,144)
(48.939347055095105,145)
(49.26639636178916,146)
(49.52147599682998,147)
(49.929107467734596,148)
(50.28367263041474,149)
};

\addplot[color=black,only marks,mark=*] coordinates
{(2.8508221333407057217,0)
(2.8508221333407065883,1)
(2.8508221333407065883,2)
(2.8508221333407065883,3)
(3.1603193508712322083,4)
(3.1603193508712431271,5)
(3.1603193508712431271,6)
(3.8490832965184452944,7)
(3.8490832965184452944,8)
(4.6539395100522323603,9)
(4.6539395100522323603,10)
(4.6539395100522323603,11)
(4.6539395100522323603,12)
(5.2709070402822467368,13)
(5.2709070402822467368,14)
(5.5659802586077403782,15)
(6.6260997035079111811,16)
(6.6546593996332280214,17)
(6.940066292118099534,18)
(7.0097767094963103427,19)
(7.0097767094963103427,20)
(7.6094611695276405063,21)
(8.4112734493786904578,22)
(8.4530249274725009952,23)
(8.9845022655779566707,24)
(8.9845022655779566707,25)
(9.1802793584589850389,26)
(9.2197399529144723175,27)
(9.2197399529144723175,28)
(9.2197399529144723175,29)
(9.8485040119965540848,30)
(9.8485040119965540848,31)
(9.8485040119965540848,32)
(10.848503905871888181,33)
(10.848503905871888181,34)
(10.848503905871888181,35)
(11.3585039332200120599,36)
(11.5253992725236930101,37)
(12.0983197213284544866,38)
(12.0983197213284544866,39)
(12.1060447011152359574,40)
(12.3096574150315424589,41)
(12.8485038711508657132,42)
(12.8485038711508657132,43)
(13.5963953167647936003,44)
(13.8485038678356297868,45)
(13.8485038678356297868,46)
(13.8485038678356297868,47)
(13.866972958354050001,48)
(14.3332697660570940672,49)
(14.8485038495813357703,50)
(14.8485038495813357703,51)
(14.8485038495813357703,52)
(15.8485039887780281071,53)
(15.8485039887780281071,54)
(15.8485039887780281071,55)
(16.4778429504912164993,56)
(16.8485034311108083102,57)
(16.8485034311108083102,58)
(16.8485034311108083102,59)
(16.8485034311108083102,60)
(17.8485056615621237818,61)
(17.8485056615621237818,62)
(17.9944122094822904225,63)
(18.8485056615588860904,64)
(18.8485056615588860904,65)
(18.8485056615588860904,66)
(18.8485056615588860904,67)
(18.9551740493845833404,68)
(19.4445285134965217773,69)
(19.8484699749574105884,70)
(19.8484699749574105884,71)
(19.8484699749574105884,72)
(20.8486127319540494699,73)
(20.8486127319540494699,74)
(20.8486127319540494699,75)
(21.5907793866693686243,76)
(21.8488983307318562436,77)
(21.8488983307318562436,78)
(21.8488983307318562436,79)
(21.8488983307318562436,80)
(22.8488983307318435501,81)
(22.8488983307318435501,82)
(22.8488983307318435501,83)
(23.8398147932422370741,84)
(23.8398147932422370741,85)
(23.8398147932422370741,86)
(24.4580615070863401049,87)
(24.4580615070863401049,88)
(25.1046658860502308291,89)
(25.1046658860502308291,90)
(25.1046667681299305633,91)
(25.1046667681299305633,92)
(25.7233683639898392058,93)
(25.7233683639898392058,94)
(26.4351809494189325246,95)
(26.7090205114511954269,96)
(26.7090205114511954269,97)
(26.7090206453120620367,98)
(27.8136010651484811861,99)
(27.8842714440118940777,100)
(27.8842714440118940777,101)
(29.2210893298290930895,102)
(29.2210893298290930895,103)
(29.2681897703421144605,104)
(29.2681897703421144605,105)
(29.4857261883127560369,106)
(29.7737169673815170287,107)
(32.1058658998966824882,108)
(32.1058658998966824882,109)
(33.4684601296397537818,110)
(34.1711649217367768468,111)
(34.1711649217367768468,112)
(34.1713157244759781948,113)
(35.4420090310400053479,114)
(35.7100213445895607537,115)
(35.7100213445895607537,116)
(37.244447499590183614,117)
(37.244447499590183614,118)
(37.244447499590183614,119)
(39.107727141675131031,120)
(39.3580075958181786,121)
(39.87011318060081796,122)
(39.87011318060081796,123)
(39.87011318060081796,124)
(40.5902580865467824,125)
(41.29996733552801846,126)
(41.29996733552801846,127)
(42.03097011588220359,128)
(42.03097011588220359,129)
(43.3075852384609,130)
(43.35303449425017,131)
(43.75899402875597,132)
(44.191360815457486,133)
(44.652410912806324,134)
(44.71790111105983,135)
(44.71790111105983,136)
(45.54664731354731,137)
(46.029027793705346,138)
(46.15930050818602,139)
(46.60025930991698,140)
(46.81567221835308,141)
(47.34461995522389,142)
(47.69375606033871,143)
(48.17668191296848,144)
(48.2344251419992,145)
(48.5197644867457,146)
(48.617489159899336,147)
(48.94246622777017,148)
(49.32442144310873,149)
};

\legend{mean,max}
\end{axis}%
\end{tikzpicture}%
\caption{\label{fig:qft-cost}Scaling of the T-count required to achieve the given approximation precision $\ve$.  This graph shows the average precision that can be achieved when using circuits with the given T-count, as well as maximal T-count required to achieve the given precision.  The results are averaged over T-optimal approximations of $R_z$ rotations by angles $\pi/2^k$ for $k=3..27$. }
\end{figure}
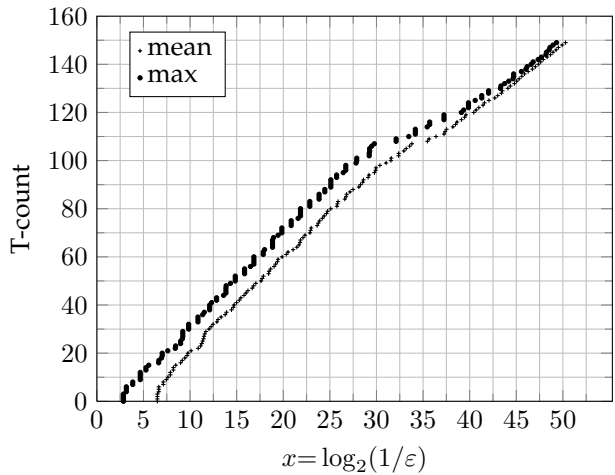

To verify correctness of our implementation, we coded a naive brute-force search algorithm that also solves Closest Unitaries Problem.  We ran both algorithms to find all optimal approximations of rotations $R_z(2\pi k/1000)$ for $k \in [0,1000]$ with at most $18$ T gates.  The two algorithms produced identical results.  The verification procedure is a part of SQCT 0.2 and can be executed via command line option~``-B''.



\section*{Acknowledgments}

Some of the authors were supported in part by the Intelligence Advanced Research Projects Activity (IARPA) via Department of Interior National Business Center Contract number DllPC20l66. The U.S. Government is authorized to reproduce and distribute reprints for Governmental purposes notwithstanding any copyright annotation thereon. Disclaimer: The views and conclusions contained herein are those of the authors and should not be interpreted as necessarily representing the official policies or endorsements, either expressed or implied, of IARPA, DoI/NBC or the U.S. Government.

This material is based upon work partially supported by the National Science Foundation (NSF), during D. Maslov's assignment at the Foundation. Any opinion, findings, and conclusions or recommendations expressed in this material are those of the author(s) and do not necessarily reflect the views of the National Science Foundation.

Michele Mosca is also supported by Canada's NSERC, MPrime, CIFAR, and CFI. IQC and Perimeter Institute are supported in part by the Government of Canada and the Province of Ontario.

We wish to thank Martin Roetteler for many helpful discussions.

\ifCLASSOPTIONcaptionsoff
  \newpage
\fi


\begin{thebibliography}{10}
\providecommand{\url}[1]{#1}
\csname url@samestyle\endcsname
\providecommand{\newblock}{\relax}
\providecommand{\bibinfo}[2]{#2}
\providecommand{\BIBentrySTDinterwordspacing}{\spaceskip=0pt\relax}
\providecommand{\BIBentryALTinterwordstretchfactor}{4}
\providecommand{\BIBentryALTinterwordspacing}{\spaceskip=\fontdimen2\font plus
\BIBentryALTinterwordstretchfactor\fontdimen3\font minus
  \fontdimen4\font\relax}
\providecommand{\BIBforeignlanguage}[2]{{%
\expandafter\ifx\csname l@#1\endcsname\relax
\typeout{** WARNING: IEEEtranS.bst: No hyphenation pattern has been}%
\typeout{** loaded for the language `#1'. Using the pattern for}%
\typeout{** the default language instead.}%
\else
\language=\csname l@#1\endcsname
\fi
#2}}
\providecommand{\BIBdecl}{\relax}
\BIBdecl

\bibitem{quant-ph/0504218}
\BIBentryALTinterwordspacing
P.~Aliferis, D.~Gottesman, and J.~Preskill, ``Quantum accuracy threshold for
  concatenated distance-3 codes,'' \emph{Quantum Information and Computation},
  vol.~6, pp. 97--165, 2006. [Online]. Available:
  \url{http://arxiv.org/abs/quant-ph/0504218}
\BIBentrySTDinterwordspacing

\bibitem{Barenco1995}
\BIBentryALTinterwordspacing
A.~Barenco, C.~Bennett, R.~Cleve, D.~DiVincenzo, N.~Margolus, P.~Shor,
  T.~Sleator, J.~Smolin, and H.~Weinfurter, ``{Elementary gates for quantum
  computation},'' \emph{Physical Review A}, vol.~52, no.~5, pp. 3457--3467,
  Nov. 1995. [Online]. Available: \url{http://arxiv.org/abs/quant-ph/9503016}
\BIBentrySTDinterwordspacing

\bibitem{BGS}
\BIBentryALTinterwordspacing
A.~Bocharov, Y.~Gurevich, and K.~M. Svore, ``Efficient decomposition of
  single-qubit gates into {$V$} basis circuits,'' \emph{Phys. Rev. A}, vol.~88,
  p. 012313, Jul 2013. [Online]. Available:
  \url{http://link.aps.org/doi/10.1103/PhysRevA.88.012313}
\BIBentrySTDinterwordspacing

\bibitem{J2}
\BIBentryALTinterwordspacing
N.~{Cody Jones}, ``{Distillation protocols for Fourier states in quantum
  computing},'' p.~18, Mar. 2013. [Online]. Available:
  \url{http://arxiv.org/abs/1303.3066}
\BIBentrySTDinterwordspacing

\bibitem{J}
\BIBentryALTinterwordspacing
N.~{Cody Jones}, J.~D. Whitfield, P.~L. McMahon, M.-H. Yung, R.~V. Meter,
  A.~Aspuru-Guzik, and Y.~Yamamoto, ``{Faster quantum chemistry simulation on
  fault-tolerant quantum computers},'' \emph{New Journal of Physics}, vol.~14,
  no.~11, p. 115023, Nov. 2012. [Online]. Available:
  \url{http://stacks.iop.org/1367-2630/14/i=11/a=115023?key=crossref.51339e3e7dfc030625bd8da9496cc34b}
\BIBentrySTDinterwordspacing

\bibitem{Coh3}
H.~Cohen, \emph{A Course in Computational Algebraic Number Theory}, ser.
  Graduate Texts in Mathematics.\hskip 1em plus 0.5em minus 0.4em\relax
  Springer, 1993.

\bibitem{Coh2}
------, \emph{Advanced Topics in Computational Number Theory}, ser. Graduate
  Texts in Mathematics.\hskip 1em plus 0.5em minus 0.4em\relax Springer New
  York, 2000.

\bibitem{Coh1}
------, \emph{Number Theory: Volume I: Tools and Diophantine Equations}, ser.
  Graduate Texts in Mathematics.\hskip 1em plus 0.5em minus 0.4em\relax
  Springer, 2007.

\bibitem{PARI}
H.~Cohen, K.~Belabas \emph{et~al.}, ``{PARI/GP}, a computer algebra system,''
  \url{http://pari.math.u-bordeaux.fr}, 1985--2013.

\bibitem{DN}
\BIBentryALTinterwordspacing
C.~M. Dawson and M.~A. Nielsen, ``{The Solovay-Kitaev algorithm},''
  \emph{Quantum Information {\&} Computation}, vol.~6, no.~1, pp. 81--95, May
  2005. [Online]. Available: \url{http://arxiv.org/abs/quant-ph/0505030}
\BIBentrySTDinterwordspacing

\bibitem{DCS}
\BIBentryALTinterwordspacing
G.~Duclos-Cianci and K.~M. Svore, ``Distillation of nonstabilizer states for
  universal quantum computation,'' \emph{Phys. Rev. A}, vol.~88, p. 042325, Oct
  2013. [Online]. Available:
  \url{http://link.aps.org/doi/10.1103/PhysRevA.88.042325}
\BIBentrySTDinterwordspacing

\bibitem{F1}
\BIBentryALTinterwordspacing
A.~G. Fowler, ``Constructing arbitrary steane code single logical qubit
  fault-tolerant gates,'' \emph{Quantum Information \& Computation}, vol.~11,
  no. 9-10, pp. 867--873, Sep. 2011. [Online]. Available:
  \url{http://arxiv.org/abs/quant-ph/0411206}
\BIBentrySTDinterwordspacing

\bibitem{arXiv:0803.0272}
\BIBentryALTinterwordspacing
A.~G. Fowler, A.~M. Stephens, and P.~Groszkowski, ``High-threshold universal
  quantum computation on the surface code,'' \emph{Phys. Rev. A}, vol.~80, p.
  052312, Nov 2009. [Online]. Available: \url{http://arxiv.org/abs/0803.0272}
\BIBentrySTDinterwordspacing

\bibitem{chem}
\BIBentryALTinterwordspacing
I.~Kassal, J.~D. Whitfield, A.~Perdomo-Ortiz, M.-H. Yung, and A.~Aspuru-Guzik,
  ``{Simulating chemistry using quantum computers.}'' \emph{Annual review of
  physical chemistry}, vol.~62, pp. 185--207, Jan. 2011. [Online]. Available:
  \url{http://www.ncbi.nlm.nih.gov/pubmed/21166541}
\BIBentrySTDinterwordspacing

\bibitem{bk:ksv}
A.~Y. Kitaev, A.~H. Shen, and M.~N. Vyalyi, \emph{{Classical and Quantum
  Computation}}, ser. Graduate studies in mathematics, v. 47.\hskip 1em plus
  0.5em minus 0.4em\relax Boston, MA, USA: American Mathematical Society, 2002.

\bibitem{arXiv:1212.0822}
\BIBentryALTinterwordspacing
V.~Kliuchnikov, D.~Maslov, and M.~Mosca, ``Asymptotically optimal approximation
  of single-qubit unitaries by {C}lifford and {T} circuits using a constant
  number of ancillary qubits,'' \emph{Phys. Rev. Lett.}, vol. 110, p. 190502,
  May 2013. [Online]. Available: \url{http://arxiv.org/abs/1212.0822}
\BIBentrySTDinterwordspacing

\bibitem{arXiv:1206.5236}
\BIBentryALTinterwordspacing
------, ``Fast and efficient exact synthesis of single-qubit unitaries
  generated by {Clifford} and {T} gates,'' \emph{Quantum Information \&
  Computation}, vol.~13, no. 7-8, pp. 0607--0630, Jul. 2013. [Online].
  Available: \url{http://arxiv.org/abs/1206.5236}
\BIBentrySTDinterwordspacing

\bibitem{MA}
\BIBentryALTinterwordspacing
K.~Matsumoto and K.~Amano, ``Representation of quantum circuits with clifford
  and $\pi/8$ gates,'' 2008. [Online]. Available:
  \url{http://arxiv.org/abs/0806.3834}
\BIBentrySTDinterwordspacing

\bibitem{bk:nc}
M.~A. {Nielsen} and I.~L. {Chuang}, \emph{Quantum computation and quantum
  information}.\hskip 1em plus 0.5em minus 0.4em\relax New York, NY, USA:
  Cambridge Univ. Press, 2000.

\bibitem{PS}
\BIBentryALTinterwordspacing
A.~Paetznick and K.~M. Svore, ``{Repeat-Until-Success: Non-deterministic
  decomposition of single-qubit unitaries},'' p.~24, Nov. 2013. [Online].
  Available: \url{http://arxiv.org/abs/1311.1074}
\BIBentrySTDinterwordspacing

\bibitem{S}
\BIBentryALTinterwordspacing
P.~Selinger, ``{Efficient Clifford+{T} approximation of single-qubit
  operators},'' Dec 2012. [Online]. Available:
  \url{http://arxiv.org/abs/1212.6253}
\BIBentrySTDinterwordspacing

\bibitem{W}
L.~Washington, \emph{Introduction to Cyclotomic Fields}, ser. Graduate Texts in
  Mathematics.\hskip 1em plus 0.5em minus 0.4em\relax Springer New York, 1997.

\bibitem{chem-small}
\BIBentryALTinterwordspacing
D.~Wecker, B.~Bauer, B.~K. Clark, M.~B. Hastings, and M.~Troyer, ``{Can quantum
  chemistry be performed on a small quantum computer?}'' p.~6, Dec. 2013.
  [Online]. Available: \url{http://arxiv.org/abs/1312.1695}
\BIBentrySTDinterwordspacing

\bibitem{WK}
\BIBentryALTinterwordspacing
N.~Wiebe and V.~Kliuchnikov, ``{Floating point representations in quantum
  circuit synthesis},'' \emph{New Journal of Physics}, vol.~15, no.~9, p.
  093041, Sep. 2013. [Online]. Available:
  \url{http://stacks.iop.org/1367-2630/15/i=9/a=093041?key=crossref.a96dd3d97b435b66541d76a6bab6451d}
\BIBentrySTDinterwordspacing

\end{thebibliography}




%
%
%
%
%
%
%
%


%
\end{document}